\documentclass[a4paper,twocolumn,11pt, accepted=2026-10-25]{quantumarticle}
\pdfoutput=1
\usepackage[colorlinks=true, citecolor=cyan, urlcolor=blue ]{hyperref}
\usepackage{epsfig,amssymb,amsmath,amscd,amsthm,amsfonts,subfigure,mathtools,braket,soul,enumitem,color,mathrsfs,dsfont} 
\usepackage{hyperref}
\usepackage{float}
\usepackage{xcolor}
\usepackage[normalem]{ulem}

\newcommand{\stkout}[1]{\ifmmode\text{\sout{\ensuremath{#1}}}\else\sout{#1}\fi}
\usepackage{array}
\usepackage{multirow}
\usepackage{graphicx}
\usepackage{cancel}
\usepackage[T1]{fontenc}
\usepackage{comment}
\usepackage[utf8]{inputenc}
\usepackage[english]{babel}
\newtheorem{definition}{Definition}
\newtheorem{theorem}{Theorem}
\newtheorem{proposition}{Proposition}
\newtheorem{corollary}{Corollary}[theorem]
\newtheorem{lemma}[theorem]{Lemma}
\newtheorem{remark}{Remark}

\newcommand{\bea}{\begin{eqnarray*}}
\newcommand{\eea}{\end{eqnarray*}}
\newcommand{\beao}{\begin{eqnarray}}
\newcommand{\eeao}{\end{eqnarray}}

\usepackage[bottom]{footmisc}
\usepackage{caption}
\newcommand{\tcr}[1]{\textcolor{red}{#1}}
\newcommand{\tcb}[1]{\textcolor{black}{#1}}

\newcommand{\RI}{\mathcal{R}}
\newcommand{\complexn}{\mathbb{C}}
\title{Sequential realization of Quantum Instruments}
\author{Soham Sau}
\orcid{0000-0002-8449-0515}
\author{Michal Sedl\'ak}
\affiliation{RCQI, Institute of Physics, Slovak Academy of Sciences, Dúbravská cesta 9, 84511 Bratislava, Slovakia}
\orcid{0000-0002-9342-5862}
\begin{document}
\maketitle
\begin{abstract}
In adaptive quantum circuits classical results of mid-circuit measurements determine the upcoming gates. This allows POVMs, quantum channels or more generally quantum instruments to be implemented sequentially, so that fewer qubits need to be used at each of the $N$ measurement steps. In this paper, we mathematically describe these problems via adaptive sequence of instruments (ASI) and show how any instrument can be decomposed into it. Number of steps $N$ and number of ancillary qubits $n_A$ needed for actual implementation are crucial parameters of any such ASI. We show an achievable lower bound on the product $N.n_A$ and we determine in which situations this tradeoff is likely to be optimal. Contrary to common intuition we show that for quantum instruments which transform $n$ to $m(>n)$ qubits, there exist $N$-step ASI implementing them just with $(m-n)$ ancillary qubits, which are remeasured $(N-1)$ times and finally used as output qubits. 
\end{abstract}

\maketitle

\section{Introduction}

Most currently developed quantum computers aim to implement quantum circuits. In this paradigm, a pure product state of $n$ qubits is prepared and transformed by a sequence of unitary gates, each acting simultaneously on only a few qubits. Finally, projective single-qubit measurements are performed on all, or on a subset, of the qubits to signal the outcome of the task.

However, for many relevant quantum information processing tasks, we start from a higher-level, often less detailed, description of what needs to be implemented. Examples include the preparation of a mixed state, the realization of a quantum channel or a positive operator-valued measure (POVM), or the implementation of a quantum instrument (i.e., a measurement having both classical and quantum output). These concepts are mathematically more convenient when solving certain optimization problems, since they focus only on the most general evolution of the primary quantum system, which can effectively arise in many different ways from interactions with other classical or quantum systems. On the other hand, they must ultimately be realized on actual quantum hardware using its available toolkit.

For example, a dilation of a quantum channel is implemented by appending ancillary qubits, performing a multi-qubit unitary gate, and then decomposing this unitary into the available quantum gates.

Current Noisy intermediate-scale quantum (NISQ) era  
devices \cite{preskill} have serious limitations especially in terms of
number of available qubits, number of gates or mid-circuit measurements that can be performed with reasonable quality within a single experimental run.
Thus, to utilize 
potential of NISQ devices efficient realizations of channels, POVMs and instruments should be developed.

Although 
single-step dilations for POVMs \cite{Naimark}, channels \cite{stinespring}, and instruments \cite{ozawa} are known, their demands for a number of ancillary qubits 
 often turn out to be too high
\cite{bharti, zoltan}. 
A practical strategy to solve this problem is to partially exchange 
space resources (qubits) for time resources (the number of computational steps).  
One of the first efforts in this direction was made by Lloyd and Viola \cite{lloyd}, who showed that 
repeated interaction with a single ancillary qubit, plus it's measurement and reset, together suffice to sequentially realize an arbitrary quantum channel. This could be considered as one of the first examples of adaptive quantum circuits. 
 
In the area of measurements, similar efforts were initiated by G. Wang and M. Ying  \cite{Ying}, who proved that extending system's Hilbert space by a single extra dimension allows for the sequential realization of any POVM. Later, Andersson and Oi \cite{AndOi} showed that using a single resettable ancillary qubit suffices to implement any POVM. 
Applicability of adaptive circuits was also considered for particular tasks, such as 
state tomography \cite{tomo2}, computing \cite{comp}, quantum sequential decoding \cite{wilde}, sequential state discrimination \cite{ssd}, joint measurability \cite{incomp} and unambiguous state discrimination \cite{danobouda}. 
Further progress was triggered by efforts to utilize sequential implementation schemes in experiments. Particularly,  Chao Shen, et.al. in \cite{channelcons}, considered such applications in superconducting circuits, including sequential implementation of quantum instruments. 
On the way towards demonstrating the practical advantage of dynamic quantum circuits containing mid-circuit measurements \cite{mcm_first} (also denoted as quantum non-demolition measurements(QND)) conditional dynamics characterization methods 
needed to be developed \cite{mcmbench,mcm_tomo_1,mcm_tomo_2, mcm_tomo_4, mcm_tomo_5, mcm_tomo_6, mcm_tomo_8}. 
In \cite{algoseq} Córcoles, et.al. implements the famous Quantum Phase estimation algorithm in its adaptive version, i.e. they  realize particular dynamic quantum circuits on a superconducting-based quantum system. They managed to obtain results comparable to a nonadaptive implementation of the same algorithm.
Later, work of P. Ivashkov, et. al. \cite{MCMexpt2024}
reported higher fidelity of implementation of four-dimensional informationally complete POVM via their hybrid sequential implementation method than via direct single-step implementation based on Naimark dilation. 
These two works clearly demonstrated that the current-era quantum technology 
could benefit from 
trading-off spatial to temporal resources, when designing the experiments.
Thus, adaptive quantum circuits are advantageous in the NISQ era experiments for implementing more complex quantum circuits thanks to utilization of mid-circuit measurements and feed-forward control \cite{entangleMCM}. 

The aim of the present work is to theoretically analyze adaptive quantum circuits in their idealized error-free form. Mid-circuit measurements are mathematically described by quantum instruments. Therefore, we begin by defining an \emph{adaptive sequence of instruments} (ASI), which we designed, from a mathematical perspective, to be a sufficiently general structure to encompass the situation under consideration. We give it a new name due to its unique way of handling classical feedforward of information. We then explore the possibilities for the sequential implementation of quantum instruments by such sequences.
\tcb{Our results generalize previous decompositions for channels, POVMs, unambiguous measurements or instruments} \cite{lloyd,AndOi,danobouda,channelcons} and we show a constructive decomposition technique for a quantum instrument into a
 $N$-step ASI. \tcb{Proposed decomposition is determined by } 
\tcb{$N-1$ stochastic matrices, which sequentially manipulate classical outcomes of the decomposed instrument.}
 We study 
\tcb{ resource 
demands} of such sequential implementation schemes 
\tcb{with respect to the dimension of a repeatedly used ancillary quantum system and the number of steps $N$. Using the proposed method, we find achievable trade-offs between these resources.}
For any instrument mapping to a Hilbert space of increasing dimensionality, we have discovered that there exists an ASI
\tcb{implementable with the smallest conceivable ancilla system.} 
\tcb{It accounts only for the increase in output dimension, and after repeated measurements in the first $N-1$ steps, it is completely consumed in the last step of the ASI.}
\tcb{The presented methods can be directly utilized for the implementation of quantum channels, POVMs, or instruments, since they provide systematic control over the number of qubits used and mid-circuit measurements.}

 The structure of the paper is organized as follows. In Section (\ref{sec1}), we review needed notation and 
 define the concept of adaptive sequence of 
 instruments. 
 
 In section (\ref{main}), we decompose an instrument into a $2$-step ASI and derive Theorem \ref{main theorem}, which plays 
 prominent role in the manuscript.
 While studying it's resource requirements we define the notion of required ancilla dimension. 
 Next, we derive Theorem \ref{th:opt2stepsplit}, which for instruments with non-decreasing Hilbert space dimension quantifies needed ancilla dimension.
For the rest of the instruments, we derive an improvement of Theorem \ref{main theorem} in section \ref{sec:imprvThm}. Section \ref{sec:n-asi} generalizes 
the above results into the $N$-step scenario. 
Considerations in Section \ref{sec:limitations_N-ASI}
 hint on scenarios in which 
 proposed decompositions might be optimal and why trade-offs valid for them might qualitatively unavoidable in general. 
 Section \ref{POVM section} specializes the presented results to the special case of sequential POVM implementation. It also hints on relation of the overall problem to instrument compatibility \cite{am1,am2,lm2}, which is further discussed together with relation to instrument postprocessing in Section \ref{sec:rel_incomp_postproc}. 
 Finally, the manuscript concludes with Section \ref{sec:summary}, which summarizes obtained results, possible future directions and open questions.

\section{Quantum Instruments and their adaptive sequences} 
\label{sec1}
Let us consider a complex Hilbert space $\mathcal{H}$ with a finite dimension $d$, and let $\mathcal{L(H)}$ denote the set of all bounded linear operators acting on $\mathcal{H}$. 
Any quantum state is represented by a density operator $\rho$ $\in \mathcal{L}\mathcal{(H)}$, i.e. a positive-semidefinite operator $(\rho \geq 0)$ with unit trace ($Tr[\rho]=1$). 
We denote the set of states by $\mathcal{S(H)}(\subset \mathcal{L(H)})$. 
An observable is described by \emph{Positive Operator-Valued Measure (POVM)} and in this paper, we consider only POVMs with finite number of outcomes. 
We denote the set of POVMs 
each having outcome set $\Omega$ and acting on Hilbert space $\mathcal{H}$  by 
$\mathcal{O}(\Omega,\mathcal{H})$. 
POVM $A \in \mathcal{O}(\Omega,\mathcal{H})$ is a collection of effects, i.e. operators $A_i \in\mathcal{L(H)}$ obeying $0\leq A_i \leq I$ $\forall i \in \Omega$, such that the normalization condition $\sum_{i \in \Omega} A_{i}=I$ holds, with $I$ being the identity operator on $\mathcal{H}$. 
The probability of observing effect $A_i$ when the system is in state $\rho$ is given by the famous Born rule as $p_i=Tr[\rho A_i]$. Indeed, in quantum theory, one can only predict the probabilities of measurement outcomes, not the individual occurrences of specific outcomes. Nevertheless, the obtained outcome can be in each round classically post-processed by stochastic transformation, effectively leading to a different overall POVM. 
This procedure is known as POVM postprocessing \cite{cleanpovm} and can be formalized as follows. 

\begin{definition}   \label{def:1}
 POVM ${B} \in \mathcal{O}(\Omega_{B},\mathcal{H})$ is called a postprocessing of POVM  ${A} \in \mathcal{O}(\Omega_{A},\mathcal{H})$, denoted by ${A} \rightarrow {B}$, if there exists a stochastic 
 matrix $\nu_{kj}$ $k\in \Omega_{A}, j \in \Omega_{B}$, i.e. $\nu_{kj}\geq 0$, $ \forall k, \forall j$, satisfying $\sum_{j \in \Omega_B} \nu_{kj}=1$ for all $k \in \Omega_A$, 
 such that
   \begin{equation}
   \label{postproc_POVM}
       {B}_j=\sum_{k \in \Omega_A} \nu_{kj}{A}_k \quad \forall j\in\Omega_B.
   \end{equation}
\end{definition}
We say observables ${A}$ and ${B}$ are \textit{postprocessing equivalent} if $A \rightarrow {B}$ and ${B} \rightarrow {A}$ , denoted by ${A} \leftrightarrow {B}$. 
If the postprocessing matrix contains only values 
zero or one,  each outcome $k$ is deterministically mapped to some outcome $j$ and pre-images (with respect to postprocessing matrix) for different values of $j$ are forming disjoint sets. 
If number of outcomes is reduced ($|\Omega_A|>|\Omega_B|$) we call such postprocessing a
coarse-graining, 
since several outcomes are joined to a single outcome and the discrimination capability of the POVM decreases. 

Next, we discuss concepts that are used to describe evolution of quantum states.
Let $\mathcal{H}$ and $\mathcal{K}$ be two Hilbert spaces. General probabilistic transformations of quantum states, are termed \textit{quantum operations}, 
and they
are described by \textit{Completely Positive Trace-Non-Increasing (CPTNI)} Maps from $\mathcal{L(\mathcal{H}})$ to $\mathcal{L(\mathcal{K})}$. 
A \textit{Quantum channel} is such a quantum operation, which 
is deterministic in nature,
and it corresponds to  
\textit{Completely Positive Trace-Preserving (CPTP)} Map from $\mathcal{L(\mathcal{H}})$ to $\mathcal{L(\mathcal{K})}$. We denote the set of quantum channels as $Ch(\mathcal{H,K})$. 
Quantum operations have a well-known representation in an operator-sum form, known as the Kraus representation \cite{Kraus}: a linear map $\mathcal{N}: \mathcal{L(\mathcal{H}}) \mapsto \mathcal{L(\mathcal{K}}) $ is a quantum operation if and only if there exists bounded operators (also called Kraus operators) $K_m$, $m=1,2,\dots, r_{\mathcal{N}}$, such that $\mathcal{N}(\rho)=\sum_{m} {K_{m}}{\rho}{K_{m}^{\dag}}$ and $\sum_m {K_m^{\dag}K_m} \leq I$. Here, the smallest achievable number $r_{\mathcal{N}} \in \mathbb{N}$ is called the Kraus rank of $\mathcal{N}$. If $\sum_m {K_m^{\dag}K_m} = I$, the operation is a quantum channel. It is a well-known fact that a composition of quantum operations is also a quantum operation. 
Suppose that the quantum system is in state $\rho\in \mathcal{S(H)}$. The probability that quantum operation $\mathcal{N}$ is realized on the system is given by $p_{\mathcal{N}}=Tr[\mathcal{N}(\rho)]$ and the normalized resulting state of the system is $\mathcal{N}(\rho)/p_{\mathcal{N}}$.
While quantum operations describe the mapping between initial and final state of the system they do not include any description of how such operation was actually implemented. 

\textit{Quantum instrument} describes a measurement procedure in such a way that it captures not only the measurement statistics, but it 
describes the post-measurement state as well. In order to achieve this, every possible outcome of an instrument is associated with a quantum operation. Thus, also quantum instruments do not include details of how measurement was performed, just its overall effect on states. 
Formally, quantum instrument can be defined \cite{Davies, measurement, ziman} as a mapping $\mathcal{I}:i\mapsto \mathcal{I}_i$ from a finite outcome set $\Omega$ to a set of quantum operations from $\mathcal{L(\mathcal{H}})$ to $\mathcal{L(K})$ such that $\phi^{\mathcal{I}}:=\sum_{i\in \Omega}\mathcal{I}_{i} \in Ch(\mathcal{H},\mathcal{K})$ is a quantum channel. This condition indeed guarantees that the total probability of getting an outcome is one, i.e. $\sum_{i\in \Omega} 
Tr[\mathcal{I}_i(\rho)]=Tr[\phi^{\mathcal{I}}(\rho)]=Tr[\rho]=1$. We will denote quantum channel $\phi^{\mathcal{I}}$ as the \emph{induced channel} of the quantum instrument $\mathcal{I}$. 
It is useful to define also the \emph{induced POVM} $A^{\mathcal{I}}\in {\mathcal{O}(\Omega,\mathcal{H})}$ of an instrument $\mathcal{I}$
by 
equation $Tr[A^{\mathcal{I}}_i\rho]=Tr[\mathcal{I}_i(\rho)]$. In the Kraus representation, one can represent the $i^{th}$ operation of the instrument as
$\mathcal{I}_i(\rho)=\sum_m K_{i,m} \rho K_{i,m}^{\dag}$ and its induced POVM can be represented as $A^{\mathcal{I}}_i=\sum_m  K_{i,m}^{\dag} K_{i,m}$. 
We denote the set of instruments, which have outcome set $\Omega$ and map from $\mathcal{L(H)}$ to $\mathcal{L(K)}$, by $Ins(\Omega,\mathcal{H},\mathcal{K})$.

Let us now discuss how quantum channels and instruments can be realized using  ancillary qubits, all initialized in state $\ket{0}$, as additional resources.
Suppose we want to realize a quantum channel $\mathcal{E}\in Ch(\mathcal{H}_2^{\otimes n},\mathcal{H}_2^{\otimes m})$ with $n$-qubit input and $m$-qubit output Hilbert space.
The key parameter for the required resources is, $r_{\mathcal{E}}$, the Kraus rank of the channel $\mathcal{E}$, since the dimension of the discarded (or ignored) subsystem (see \cite{ziman}) needs to be at least $r_{\mathcal{E}}$. Consequently, the number of qubits at the end of a channel dilation must be at least $m+\lceil \log r_{\mathcal{E}} \rceil$. Thus, when starting the dilation we need $\max{\{0,m+\lceil \log r_{\mathcal{E}} \rceil-n}\}$ ancillary qubits. Channel $\mathcal{E}$ is then achieved by performing suitable unitary transformation on the system plus ancillary qubits and ignoring the suitable part of such extended system afterwards.

Similarly, if we want to realize quantum instrument $\mathcal{I}\in Ins(\Omega,\mathcal{H}_2^{\otimes n},\mathcal{H}_2^{\otimes m})$ just by: i) adding ancillary qubits, ii) unitary evolution and  iii) projective measurements on some of the qubits, then (assuming finite outcome space $\Omega$) what matters is, $r_{\mathcal{I}}=\sum_{i\in \Omega} r_i$, the sum of the Kraus ranks of the quantum operations $\mathcal{I}_i$ defining the instrument. Analogically, in this case we need $\max{\{0,m+\lceil \log r_{\mathcal{I}} \rceil-n}\}$ ancillary qubits.

Quite important example of an instrument, which we will often utilize also in this work, is a Lüders instrument \cite{luders}.
\begin{definition} \label{luders def}
Lüders instrument $\mathcal{I}^A \in Ins(\Omega,\mathcal{H},\mathcal{H})$ of POVM $A\in \mathcal{O}(\Omega,\mathcal{H})$ is defined as
\begin{align}
\mathcal{I}^A_i(\rho)=\sqrt{A_i}\rho \sqrt{A_i} \quad \forall \rho \in \mathcal{S(H)}, \;\forall i\in \Omega.
\end{align}
\end{definition}
Instrument $\mathcal{I}^A$ represents the least disturbing (see e.g. \cite{hayashi1}) way of measuring the POVM $A$ from the point of view of the introduced state change. It's Kraus rank, $r_{\mathcal{I}^A}=\sum_{i\in \Omega} 1=|\Omega|$, is the number of elements of the outcome set $\Omega$. This minimal possible Kraus rank  $r_{\mathcal{I}^A}$ reflects in small (smallest possible in a sense) requirement on the number of ancillary qubits.

The concept of a quantum instrument is quite versatile, as it can also represent a quantum state, a POVM, or a channel as a special case. In particular, instruments with a single- element outcome set uniquely represent channels. On the other hand, instruments with a one-dimensional output space are in one-to-one correspondence with POVMs. Finally, instruments with a one-dimensional input space and a single- element outcome set are equivalent to quantum states on the output Hilbert space. 

The composition of two quantum instruments is well defined if the output Hilbert space of the first is the same as the input Hilbert space of the second. The resulting object is again a quantum instrument, with an outcome space given by the Cartesian product of the original outcome spaces.

Note that the induced POVM and the induced channel are unique to a given instrument. However, different instruments can have the same induced POVM, the same induced channel, or even both. 

If a system undergoes a sequence of measurements or any other form of evolution, then the composition of instruments provides the most general description of this situation. Suppose an instrument $\mathcal{I}$ is performed on the system. We may ask what kind of interventions we can perform to obtain a different overall instrument $\mathcal{J}$ describing the total evolution of the system from beginning to end. The following definition was introduced in \cite{lm1}.

\begin{definition}
\label{def:postproc}
We say an instrument $\mathcal{J} \in Ins(\Omega,\mathcal{H},\mathcal{K})$ can be postprocessed to instrument
$\mathcal{T}\in Ins(\Lambda,\mathcal{H},\mathcal{V})$ and we denote it as $\mathcal{J}\rightarrow \mathcal{T}$ if there exist instruments $\mathcal{R}^{i} \in Ins(\Lambda,\mathcal{K},\mathcal{V})$ $\forall i\in\Omega$ such that $\mathcal{T}_j=\sum_{i\in\Omega} \mathcal{R}^i_j \circ \mathcal{J}_i$ $\forall j\in \Lambda$.
\end{definition}

Instrument postprocessing introduces a partial order among equivalence classes of instruments. If $\mathcal{J} \rightarrow \mathcal{T}$ and $\mathcal{T} \rightarrow \mathcal{J}$ hold simultaneously, we say that instruments $\mathcal{J}$ and $\mathcal{T}$ are postprocessing equivalent.

In the rest of the manuscript, we will often consider similar compositions of more than two adaptively chosen instruments, so it is useful to formalize such structures.

\subsection{Adaptive Sequence of Quantum Instruments (ASI)}

Our aim is to define an $N$-step \emph{adaptive sequence of instruments} $\mathcal{Q}$ in which the choice of each instrument in the sequence depends directly on the classical outcome of the preceding instrument. At first glance, it might seem restrictive to have dependence only on the immediately preceding outcome. However, we will show later that arbitrary types of outcome dependence can be recast in this way through suitable choices of instruments and their outcome sets.

Thus, we propose the following definition. 
\begin{definition}
\label{def:ASI}
Let $\{\mathcal{H}_0,\cdots,\mathcal{H}_N\}$ be an $(N+1)-$tuple of 
Hilbert spaces, $\Omega_0\equiv\{1\}$ and $\{\Omega_1,\cdots,\Omega_N\}$ be a $N-$tuple of outcome spaces.
We denote intermediate outcome at step $k$ as $a_k$ and we use upper index \tcb{to} signalize step number and to distinguish different instruments.
Thus, let us 
choose sets of instruments 
$\mathcal{Q}^{k} \equiv \{\mathcal{I}^{k,a_{k-1}}\}_{a_{k-1} \in \,\Omega_{k-1}}\subset Ins(\Omega_k, \mathcal{H}_{k-1},\mathcal{H}_k)$ $\forall k$, $1\leq k \leq N$.
The $N$-step adaptive sequence of instruments $\mathcal{Q}$ is an 
$N$-tuple of all sets $\mathcal{Q}^{k}$, 
i.e. 
$\mathcal{Q}\equiv ({\mathcal{Q}^1,\cdots \mathcal{Q}^N})$.
\end{definition}

One can see a composition of $N$ fixed instruments $\mathcal{I}^{1},\ldots, \mathcal{I}^{N}$ as a special case of 
ASI. This situation corresponds to choosing all the elements of the set $\mathcal{Q}^k$ the same, i.e., $\mathcal{I}^{k,a_{k-1}}=\mathcal{I}^{k}$ $\forall a_{k-1} \in \Omega_{k-1}$, $\forall k=1,\ldots,N$. 

Let us introduce the Total Instrument by the following definition. 

\begin{definition}
\label{def:TotalI}
Let $\mathcal{Q}$ be $N-$step adaptive sequence of instruments. Then we define, a \emph{Total Instrument} $\mathcal{T}$ $\in Ins(\Omega_N,\mathcal{H}_0,\mathcal{H}_N)$ of $\mathcal{Q}$,  as 
\begin{equation}
      \label{Total instrument defined}
    \mathcal{T}_{a_N}=\sum_{a_{N-1}\in\Omega_{N-1}}\!\cdots\sum_{a_1 \in \Omega_{1}}^{} \biggl( \mathcal{I}^{N,a_{N-1}}_{a_{N}} \circ \cdots \circ \mathcal{I}^{1,1}_{a_1} \biggl) \, 
\end{equation}
   
\end{definition}

Total Instrument describes the 
overall instrument corresponding to the adaptive sequence of instruments $\mathcal{Q}$ if we forget (or ignore) the intermediate classical outcomes $a_1,\ldots,a_{N-1}$. This approach is useful when we want the adaptive sequence of instruments to (deterministically) perform a fixed instrument with outcome set $\Omega_N$. 
Thus, the Total Instrument $\mathcal{T}$ has the same outcome space as instruments in $\mathcal{Q}^N$ and 
it slightly generalizes the idea of instrument postprocessing 
(see Definition \ref{def:postproc}) 
into a multi-step scenario.

\begin{figure}
\centering
  \includegraphics[width=8.2cm]{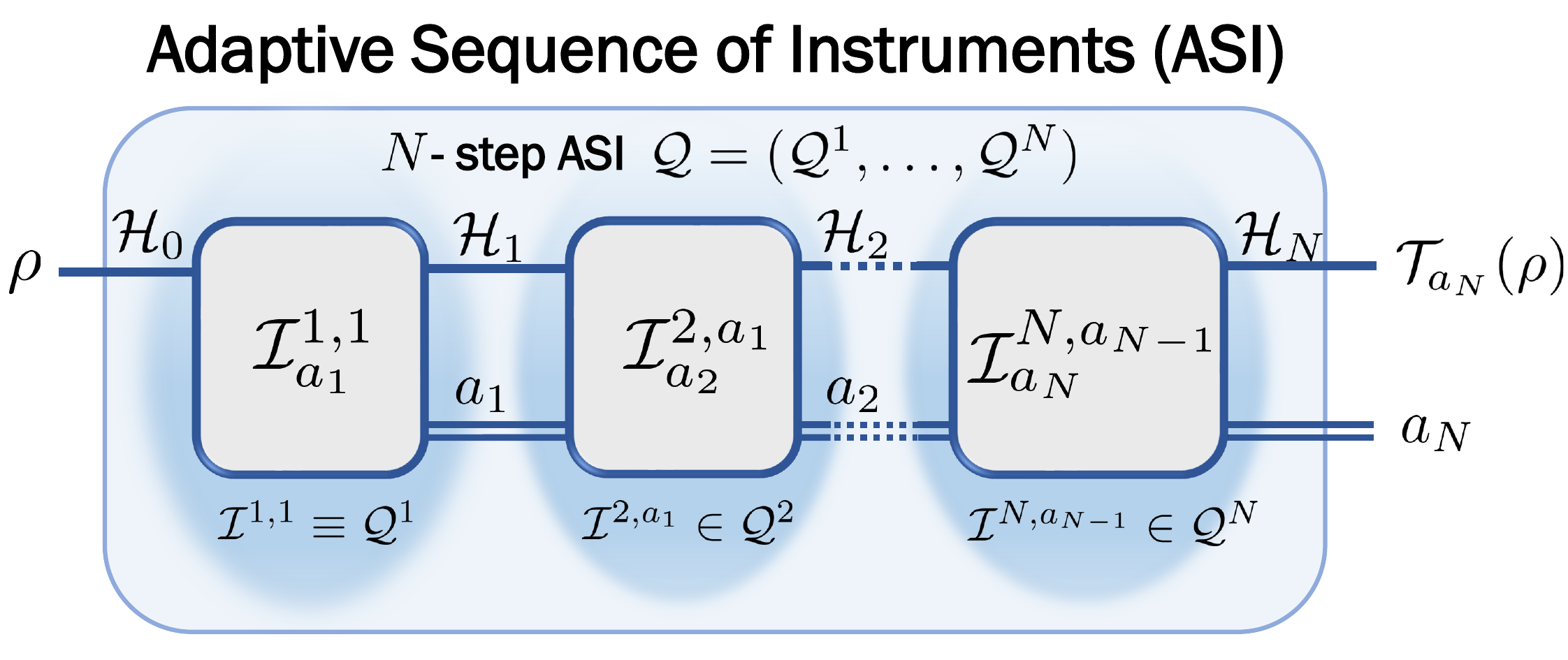}
  \captionof{figure}{Illustration of $N$-step Adaptive sequence of instruments defined in Definition \ref{def:ASI}.Quantum instrument in the $k$-th step is chosen based on the outcome $a_{k-1}$ obtained in the previous step/steps. For explanation on generality of outcome dependence see discussion below Remark \ref{remarkASI1} in the main text.}
 \end{figure}

For further work 
it will be useful to define terms describing parts of the sequence of instruments which are split into the first $N-1$ steps and the last step.

\begin{definition}
For 
ASI $\mathcal{Q}$, we define the Initial Instrument 
$\mathcal{J}\in Ins(\Omega_{N-1},\mathcal{H}_0,\mathcal{H}_{N-1})$ as 
the Total Instrument of the $N-1$ step ASI 
$({\mathcal{Q}^1,\ldots ,\mathcal{Q}^{N-1}})$, i.e. we have by definition quantum operations 
$\mathcal{J}_{a_{N-1}}=\sum_{a_{N-2}\in\Omega_{N-2}}\!\cdots\sum_{a_1 \in \Omega_{1}}^{} \mathcal{I}^{N-1,a_{N-2}}_{a_{N-1}} \circ \cdots \circ \mathcal{I}^{1,1}_{a_1}$,
and we define the set of Residual Instruments, 
denoted by $\RI^{a_{N-1}}:=\mathcal{I}^{N,a_{N-1}} \in Ins(\Omega_N,\mathcal{H}_{N-1},\mathcal{H}_N)$ such that
\begin{equation} \label{eq:TTTT}
     \mathcal{T}_{a_N} =\sum_{a_{N-1}\in \Omega_{N-1}} \RI^{a_{N-1}} \circ \mathcal{J}_{a_{N-1}}
\end{equation} 
\end{definition}

\begin{remark}
\label{remarkASI1}
It is easy to notice that any $2$-step ASI $\mathcal{Q}$ can be also seen as a particular realization of an instrument postprocessing relation ($\mathcal{I}^{1,1}\rightarrow \mathcal{T}$) of initial instrument $\mathcal{I}^{1,1}$ of $\mathcal{Q}$ to it's Total instrument $\mathcal{T}$. Indeed, by Definition \ref{def:TotalI} we have
$\mathcal{T}_{a_2}=\sum_{a_{1}\in\Omega_{1}} \mathcal{I}^{2,a_{1}}_{a_{2}} \circ \mathcal{I}^{1,1}_{a_1}$ $\forall a_2\in\Omega_2$, which coincides with Definition \ref{def:postproc}.
\end{remark}

Let us now return to a possible multi-step dependence of the instrument choice. 
We will illustrate the in-built flexibility of the 
ASI in a three-step scenario, but the method works analogically in a general case. 
Suppose that in an experiment we perform first instrument $\mathcal{E}\in Ins(\Lambda_1,\mathcal{H}_0,\mathcal{H}_1)$ obtaining outcome $b_1$. Based on $b_1$ we decide to apply instrument $\mathcal{F}^{b_1}\in Ins(\Lambda_2,\mathcal{H}_1,\mathcal{H}_2)$ and obtain outcome $b_2$. Finally, based on outcomes $b_1$ and $b_2$ (i.e. based on $(b_1,b_2)$) we carry out $\mathcal{G}^{(b_1,b_2)}\in Ins(\Lambda_3,\mathcal{H}_2,\mathcal{H}_3)$ and obtain outcome $b_3$. We see that the choice of the last instrument depends on the previous two results. Our goal now is to describe how this situation can be equivalently described by the $3$ step 
ASI. 
We choose $\Omega_1=\Lambda_1$, $\Omega_2=\Lambda_1 \times \Lambda_2$, $\Omega_3=\Lambda_3$. 
We define 
for all
$a_1\in \Omega_1$, $a_2=(b_1,b_2) \in \Omega_2$, $a_3\in \Omega_3$
\begin{align}
\mathcal{I}^{1,1}_{a_1}&=\mathcal{E}_{a_1} \\
\mathcal{I}^{2,a_1}_{(b_1,b_2)}&=
  \left \{
                         \begin{aligned}
                           & \mathcal{F}^{b_1}_{b_2}
                           && \mbox{ if } b_1=a_1 \\
                           & 0 
                           && \mbox{ if } b_1 \neq a_1
                        \end{aligned}  \right.\\
  \nonumber \\
  \mathcal{I}^{3,(b_1,b_2)}_{a_3}&= \mathcal{G}^{(b_1,b_2)}_{a_3} \;.
\end{align}
We see that by enlarging the outcome set of the instruments 
in the $k$-th step these instruments can be always defined in such a way that their outcomes encode obtained outcomes of (potentially) all previous steps, thus allowing for arbitrary dependence on the previous outcomes.

\section{Sequential realization of quantum instruments} \label{main}

Suppose we obtain instrument $\mathcal{T}$ as the composition of instruments $\mathcal{J}$ and $\RI$. If we are asked to implement instrument $\mathcal{T}$, we can either use a dilation technique directly for it or construct a different dilation by combining dilations of instruments $\mathcal{J}$ and $\RI$. This might be useful if we want to minimize the use of resources such as ancillary systems. In particular, if we can reset (reinitialize) the ancillary system, then the same qubits can be used to realize the dilation of instrument $\mathcal{J}$ and of instrument $\RI$. In this way, one can trade spatial resources (the number of qubits) for temporal resources (the number of steps required to realize the task), since the dilations of $\mathcal{J}$ and $\RI$ would be carried out one after the other. 
The example above shows that realizing an instrument sequentially can be useful for its practical implementation. The goal of this section is to explore how a fixed instrument $\mathcal{T}$ can be realized by an adaptive sequence of instruments. We will also refer to $\mathcal{T}$ as the Total Instrument. We can also say that we want to decompose the Total Instrument into an ASI. 

\tcb{
Before we state our main result, let us briefly discuss basic qualitative properties, which such decomposition must obey. It is known that there is an information gain vs. state disturbance tradeoff  \cite{disturbvsinfo}, when we perform a measurement, or more precisely a quantum instrument. As a simple consequence, if we want to decompose 
an isometric channel (i.e. one outcome instrument with just a single Kraus operator) 
then it implies that first step of ASI must not make any (non-correctable) disturbance to the states. 
Consequently, an induced POVM of the first instrument has to be trivial POVM, whose effects do not reveal any information about the measured state, since they are proportional to unity (this must actually hold for all instruments of the considered ASI). 
}

\tcb{
Similarly, it is known that some quantum operations are not deterministically invertible. Thus, in general the first step of ASI can not perform such disturbances, which we are not able to 
later modify to reach 
the desired Total instrument. 
On the mathematical level it is suitable to realize that trace preservation and complete positivity implies that singular values of Kraus operators are between zero and one \cite{singularkraus}. Thus, if want to deterministically invert action of a single Kraus operator we often need singular values larger than one, which 
do not exist in quantum operations.  
This problem 
can be avoided in sequential implementation if we aim at further refinement of the quantum operation instead of its inversion. 
In a sense the quantum operation and the classical information gained in the first step 
has to be compatible with final quantum operation of the Total instrument we end up implementing in the given run of the experiment.
}
\tcb{
For example, performing projective measurement of a system (understood here as a Lüders instrument defined by projection valued 
measure) as the first step of ASI will 
turn out to be unsuitable for reaching Total instrument, whose induced POVM is a different projection valued measure or a general POVM.
}

%

The following theorem shows \tcb{sufficient} conditions under which we can 
always decompose the Total Instrument into a $2$-step ASI. 
\tcb{We will address necessity of the conditions in Section \ref{sec:rel_incomp_postproc}.}

\begin{theorem} \label{main theorem}
Given a Total Instrument $\mathcal{T}$ $\in Ins(\Omega,\mathcal{H},\mathcal{K})$ and any POVM $B \in \mathcal{O}(\Omega_B,\mathcal{H})$ such that the induced POVM of $\mathcal{T}$ can be postprocessed to POVM $B$ 
(i.e. $A^{\mathcal{T}}\rightarrow B$), then there exists an adaptive sequence of an Initial Instrument $\mathcal{J}$ $\in Ins(\Omega_B,\mathcal{H},\mathcal{H})$ and a set of Residual Instruments $\RI^{j} \in Ins(\Omega,\mathcal{H},\mathcal{K})$ $\forall j \in \Omega_B$ such that
\begin{equation} \label{eq:totallll}
    \mathcal{T}_k= \sum_{j\in \Omega_B} \RI^{j}_k \circ \mathcal{J}_{j} \quad \quad\forall k\in \Omega
    \end{equation}
   
where the Initial Instrument $\mathcal{J}_{j}(\rho)=\sqrt{B_{j}}\rho\sqrt{B_{j}}, \forall \rho \in \mathcal{S(H)}$ is the Lüders instrument of POVM $B$.
 
\end{theorem}

\begin{proof}
Due to postprocessing relation $A^{\mathcal{T}} \rightarrow B$ (see Eq. \ref{postproc_POVM}) there exist postprocessing matrix $\nu_{kj}$ $k\in \Omega, j \in \Omega_{B}$.
Let $\{{T^{\prime}}_{k,m}\}$ be the set of the Kraus operators of the Total instrument, 
i.e. 
\begin{equation}
    \mathcal{T}_{k}(\rho)= \sum_{m}{{T}^{\prime}}_{k,m}\;\rho\;{{T}^{\prime}}_{k,m}^{\dag} 
\end{equation}
where $\sum_{k\in \Omega} \sum_{m}{{T}^{\prime}}_{k,m}^{\dag}{T^{\prime}}_{k,m}=I_\mathcal{H}$. 
We choose another set of Kraus operators $\{{T}_{k,jm}\}$ of $\mathcal{T}$, such that the two sets of Kraus operators are related by 
\begin{align}
\label{def:Kkjm}
{T}_{k,jm}= \sqrt{\nu_{kj}}\;{T}^{\prime}_{k,m}
\end{align}
where 
$k\in \Omega, j\in\Omega_B$, and $m=1,\ldots,r_{\mathcal{T}_k}$. 
It is straightforward to verify that, indeed, the action of both sets of Kraus operators is the same, i.e. 
\begin{align}
\label{eq:Tkrho}
\mathcal{T}_{k}(\rho) &= \sum_{j,m}{{T}}_{k,jm}\;\rho\; {{T}}_{k,jm}^{\dag} \quad \forall k\in \Omega, \rho \in \mathcal{S(H)} 
\end{align}
and the normalization condition is as well satisfied
\begin{align}
\label{eq:normt}
\sum_{k,j,m} &{{T}_{k,jm}}^{\dag}{{T}_{k,jm}}= \sum_{k,m} \;\sum_j \nu_{kj}\;{{T}^{\prime}}_{k,m}^{\dag} {T^{\prime}}_{k,m}=I_\mathcal{H} 
\end{align}
We observe that using equations (\ref{postproc_POVM}) and (\ref{def:Kkjm}) POVM elements $B_j$ can be expressed via Kraus operators $T_{k,jm}$ as
\begin{align}
\label{eq:BjbyT}
B_j=\sum_{k,m} \nu_{kj} {T}^{\prime\,\dag}_{k,m} {T}^\prime_{k,m}= \sum_{k,m} {T}^{\dag}_{k,jm} {T}_{k,jm} .
\end{align}
Let us define a projector 
$\Pi_j$ 
onto the support of operator $B_j$.
We denote the Moore-Penrose pseudo-inverse \cite{moo, pen}  of $B_{j}^{\frac{1}{2}}$ by $B_{j}^{-\frac{1}{2}}$. Thus, 
\begin{equation} \label{projector}
   {\Pi}_{j}B_{j}^{\pm\frac{1}{2}}=B_{j}^{\pm\frac{1}{2}}.
\end{equation} 

We can represent operators ${T}_{k,jm}$ in singular value decomposition 
\begin{equation} \label{singular}
    {T}_{k,jm}=\sum_{v}\lambda_v^{k,jm} \ket{f_v^{k,jm}} \bra{e_v^{k,jm}}
\end{equation}
where for every $k,j,m$, vectors $\ket{e^{k,jm}_v}, \ket{f^{k,jm}_v}$,where $v=1,\ldots r_{k,jm}$ are chosen from suitable orthonormal bases $\{\ket{e^{k,jm}_v}\}^{d_{\mathcal{H}}}_{v=1}$ and 
$\{\ket{f^{k,jm}_v}\}^{d_{\mathcal{K}}}_{v=1}$ in 
$\mathcal{H}$ and $\mathcal{K}$ 
respectively, and $\lambda_v^{k,jm}$ are 
positive numbers. Using Eq. (\ref{eq:BjbyT}) we can write
$ B_{j}=\sum_{k,m,v}({\lambda_v^{k,jm}})^2 \ket{e_v^{k,jm}} \bra{e_v^{k,jm}}$,
which can be written as the sum of two positive semidefinite operators, i.e.
\begin{align}\label{eq:Bj_as_sum} 
 B_{j}=& \left( \sum_{v, k',m'}({\lambda_v^{k',jm'}})^2 \ket{e_v^{k',jm'}} \bra{e_v^{k',jm'}} \right. \nonumber\\
&\left. -\sum_{v}({\lambda_v^{k,jm}})^2 \ket{e_v^{k,jm}} \bra{e_v^{k,jm}} \right)  \\
&+ \sum_{v}({\lambda_v^{k,jm}})^2 \ket{e_v^{k,jm}} \bra{e_v^{k,jm}} \nonumber
\end{align}
We can see that the support of $\sum_{v}({\lambda_v^{k,jm}})^2 \ket{e_v^{k,jm}} \bra{e_v^{k,jm}}$ is same as the input range of ${T}_{k,jm}$ which means the input range of ${T}_{k,jm}$ is included in 
the support of $B_j$ . (See Appendix \ref{appendix1} for the proof). For the next steps, this finding is most conveniently formulated for us as  
\begin{align}
\label{eq:tkjm_res}
   {T}_{k,jm}\Pi_{j}= {T}_{k,jm}
\end{align}

Next, let us first consider the case of $\dim \mathcal{H} \leq \dim \mathcal{K}$. In this case it is always possible to define the Residual instruments, $\RI^j$, using the Kraus operators of the Total instrument and without introducing any auxiliary ones. Let us make the following choice of Kraus operators ${R}_{k,m}^{j}$ 
\begin{align} 
\label{eq:RES1}
{R}^{j}_{k,m}& ={T}_{k,jm}B_j^{-\frac{1}{2}} + c_{k,jm} V_{k,jm}(I_{\mathcal{H}}-\Pi_j),
\end{align}
where 
$c_{k,jm}$ is any set of complex numbers, such that $\sum_{k,m} |c_{k,jm}|^2=1$ and $V_{k,jm}=\sum_{v=1}^{\dim \mathcal{H}} \ket{f_v^{k,jm}} \bra{e_v^{k,jm}}$ is an isometry from $\mathcal{H}$ to $\mathcal{K}$.

Valid quantum instruments must satisfy the normalization condition. This is expressed for Total instrument $\mathcal{T}$ by Eq. (\ref{eq:normt}). Lüders instrument $\mathcal{J}$ is normalized by construction. Thus, it remains to prove that Kraus operators of residual instruments are properly normalized.

For every instrument $\RI^j$ we get
\begin{align}
\label{eq:normRjkm}
\sum_{k,m} R^{j\,\dag}_{k,m} \, R^j_{k,m}  
=&  \sum_{k,m} {B_{j}}^{-\frac{1}{2}} {T}_{k,jm}^{\dag} {T}_{k,jm} {B_{j}}^{-\frac{1}{2}} \nonumber\\
     &+ (I_{\mathcal{H}}-{\Pi}_{j})(I_{\mathcal{H}}-{\Pi}_{j}) \nonumber \\
=&\; \Pi_j+ (I_{\mathcal{H}}-{\Pi}_{j})=I_{\mathcal{H}},
\end{align}
where we used Eq. (\ref{eq:BjbyT}), $V^\dagger_{k,jm}V_{k,jm}=I_{\mathcal{H}}$ and the fact that 
support of an operator $\ket{e_v^{k,jm}} \bra{e_v^{k,jm}}$ 
is included in the support of $B_j$ (proof is immediate from Eq. (\ref{eq:Bj_as_sum}) and appendix  \ref{appendix1}). 

Next, using equations (\ref{eq:tkjm_res}) and (\ref{eq:RES1}) and the definition of $\mathcal{J}$ from Theorem \ref{main theorem}, we get 
\begin{align} 
    \sum_{j \in \Omega_B} (\RI^{j}_k \,\circ \,\mathcal{J}_{j})(\rho) 
    =& \sum_{j,m} {T}_{k,jm}B_{j}^{-\frac{1}{2}}B_{j}^{\frac{1}{2}}\rho B_{j}^{\frac{1}{2}}B_{j}^{-\frac{1}{2}} {T}_{k,jm}^{\dag}   \nonumber\\
    =& \sum_{j,m} {T}_{k,jm}\Pi_{j} \,\rho\, \Pi_{j} {T}_{k,jm}^{\dag}  \nonumber \\
    =& \sum_{j,m} {T}_{k,jm}\,\rho\, {T}_{k,jm}^{\dag}   \label{biggy}
\end{align}
$\forall \rho \in \mathcal{S(H)}$, where other potential terms in the second line vanished due to Eq. (\ref{projector}).
Thus, equation (\ref{biggy}) equals $\mathcal{T}_k(\rho)$. Hence, we have proved that
$\mathcal{T}_{k}= \sum_{j \in \Omega_B} \RI^{j}_k \circ \mathcal{J}_{j}$ in the case of $\dim \mathcal{H} \leq \dim \mathcal{K}$.

Next, we analyze the case of  $\dim \mathcal{H} > \dim \mathcal{K}$. 
In general, this case is more complex, because there might not be enough space to accommodate subspace
$I_{\mathcal{H}}-{\Pi}_{j}$ in the output Hilbert space $\mathcal{K}$. In such cases, we may need one or even several additional Kraus operators. Here, we present a generally valid solution, but a point to be noted is that the solution may not always be resource efficient 
in terms of the total number of Kraus operators. After the proof, we discuss its potential improvements. Let us define Hilbert space denoted by $(I_{\mathcal{H}}-{\Pi}_{j})\mathcal{H}$ as orthocomplement of the support of operator $B_j$ in Hilbert space $\mathcal{H}$.
For each $j\in\Omega_B$ let us choose some quantum instrument $\RI'^j\in Ins(\Omega,(I_{\mathcal{H}}-{\Pi}_{j})\mathcal{H},\mathcal{K})$ preferably with lowest possible Kraus rank. 
Thus, the normalization of these instruments reads
\begin{align}
    \sum_{k\in \Omega} \;\;\sum_{m=r_{k}+1}^{r_k+r'^j_k} (R'^j_{k,m})^\dagger R'^j_{k,m} = I_{\mathcal{H}}-{\Pi}_{j} \quad \forall j\in \Omega_B,
\end{align}
where we choose the index $m$ to start after $r_k$, which is the Kraus rank of $\mathcal{T}_k$, and $r'^j_k$ marks the Kraus rank of quantum operations $\RI'^j_k$. 
Let us now define
\begin{align} 
\label{eq:def_RESIDUAL_2}
{R}^{j}_{k,m}=
    \begin{cases}
        {T}_{k,jm}B_j^{-\frac{1}{2}} & m=1,\ldots,r_k \\
        R'^j_{k,m} & m=r_{k}+1, \ldots, r_k+r'^j_k
    \end{cases}    
\end{align}
We can now easily verify validity of Eqs. (\ref{eq:normRjkm}) and (\ref{biggy}), which is obvious from the construction of $\RI'^j$ and Eq. (\ref{eq:tkjm_res}).

Thus, we have realized the Total Instrument $\mathcal{T}$ as a $2$-step 
ASI $\mathcal{Q}=(\mathcal{Q}^1,\mathcal{Q}^2)$ constituted by Initial Instrument  $\{\mathcal{J}\}=\mathcal{Q}^1$ and a set of Residual Instruments $\{\RI^j, j\in \Omega_B\}=\mathcal{Q}^2$.
This completes the proof. 
\end{proof}


\tcb{
Constructive proof of Theorem \ref{main theorem} (presented above) enables us to derive the following corollary describing Kraus ranks of residual instruments generated by Theorem \ref{main theorem}.
}
\tcb{
\begin{corollary}
\label{cor:KrausRanks}
Suppose Theorem \ref{main theorem} was applied to Total Instrument $\mathcal{T}$ $\in Ins(\Omega,\mathcal{H},\mathcal{K})$ 
using POVM $B \in \mathcal{O}(\Omega_B,\mathcal{H})$ such that $A^{\mathcal{T}}\rightarrow B$. 
If $\dim \mathcal{H} \leq \dim \mathcal{K}$ then Kraus rank of each residual instrument $\mathcal{R}^j$
is the sum of Kraus ranks of the quantum operations of the Total instrument, which it realizes 
(${\rm rank}(\mathcal{R}^j)=\sum_{k:\,\nu_{kj}>0} {\rm rank}(\mathcal{T}^k)$). If $\dim \mathcal{H} < \dim \mathcal{K}$ then 
${\rm rank}(\mathcal{R}^j)$ is larger than the above value by at least 
\begin{align}
\label{eq:addKraus}
n_{add}=\max[0,    \lceil &\frac{1}{d_{\mathcal{K}}}(Tr(I_{\mathcal{H}} 
-{\Pi}_{j})\nonumber\\
&-\!\!\!\sum_{m,k:\,\nu_{kj}>0}(d_{\mathcal{K}}-d^j_{k,m}))\rceil], 
\end{align}
where $\Pi_j$ is the projector onto the support of $B_j$,  $d^j_{k,m}=\dim Im(\sqrt{\nu_{kj}}\;{T}^{\prime}_{k,m})$ and ${T}^{\prime}_{k,m}$ are Kraus operators of quantum operation $\mathcal{T}_{k}$.
\end{corollary}
}

\begin{proof}
\tcb{
If $\dim \mathcal{H} \leq \dim \mathcal{K}$ then
from the details of the proof of Theorem \ref{main theorem}, especially from Eqs. (\ref{def:Kkjm}) and (\ref{eq:RES1}) it follows that the number of non-zero Kraus operators of Residual instrument $\RI^j$ can be chosen to be $\sum_{k:\,\nu_{kj}>0} {\rm rank}(\mathcal{T}^k)$.
}
\tcb{On the other hand, suppose $\dim \mathcal{H} > \dim \mathcal{K}$ and} we would like to optimize the number of Kraus operators for each instrument $\mathcal{R}^j$.
Ideally, we would introduce similar terms as in Eq. (\ref{eq:RES1}), which is used when $\dim \mathcal{H} \leq \dim \mathcal{K}$, to avoid introducing additional Kraus operators. However, the image of ${T}_{k,jm}$ already occupies $d^j_{k,m}=\dim Im({T}_{k,jm})$ dimensions. Thus, a given Kraus operator $R^j_{k,m}$ can only faithfully 
transfer 
additional $d_{\mathcal{K}}-d^j_{k,m}$ dimensions. Out of such pieces of $R^{j\,\dag}_{k,m} \, R^j_{k,m}$ we need to build up $(I_{\mathcal{H}}-{\Pi}_{j})$, the missing part of the normalization. Since we know that the rank of a sum of positive operators is, at most the sum of ranks of the individual operators, such pieces together can create at most an operator with a rank $\sum_{k:\,\nu_{kj}>0}\tcb{\sum_m}(d_{\mathcal{K}}-d^j_{k,m})$. We remark here that if $\nu_{kj}=0$ then ${T}_{k,jm}=0$ \tcb{$\forall m$} and such an operator does not need to be implemented. Thus, adding some additional part to such ${T}_{k,jm}$ also means introducing additional Kraus operator(s).
We can summarize that if $\sum_{k:\,\nu_{kj}>0}\tcb{\sum_m}(d_{\mathcal{K}}-d^j_{k,m})\geq Tr(I_{\mathcal{H}}-{\Pi}_{j})$ then no additional Kraus operators need to be introduced for a given instrument $\RI^j$. 
\tcb{
Otherwise, 
the number of additional Kraus operators is given by 
\begin{align}
\lceil \frac{1}{d_{\mathcal{K}}}(Tr(I_{\mathcal{H}}-{\Pi}_{j})-\sum_{k:\,\nu_{kj}>0}\tcb{\sum_m}(d_{\mathcal{K}}-d^j_{k,m}))\rceil , \nonumber
\end{align}
which concludes the proof.
}
\end{proof}


\begin{remark}
In Theorem \ref{main theorem} the Initial instrument $\mathcal{J}$ depends only on the POVM $B$, which is a postprocessing of the induced POVM of the Total instrument $A^{\mathcal{T}}$. Thus, we see that the actual state change introduced by the Total instrument does not influence instrument $\mathcal{J}$. In conclusion, all total instruments with the same induced POVM $A^{\mathcal{T}}$ can use the same initial instrument.

\end{remark}

\begin{remark}
\label{re_outcomes}
Let us note that outcomes $j$ and $k$ obtained in the sequential realization of the instrument  $\mathcal{T}$  from Theorem \ref{main theorem} are not completely independent. By the proof of Theorem \ref{main theorem}, especially its Eq. (\ref{def:Kkjm}), we see that $k$ can only take such values, for which the postprocessing matrix elements $\nu_{k,j}>0$. In other words, the sequential realization is designed in such a way that some of the outcomes $k\in\Omega$ of the residual instrument $\RI^j$ have zero probability of appearance for all states produced by quantum operation $\mathcal{J}_j$ of the Initial Instrument $\mathcal{J}$. \\  
\end{remark}

\subsection{Example illustrating Remark \ref{re_outcomes}} 

This subsection shows a simple example of the use of Theorem \ref{main theorem}. 
Let us consider a three outcome POVM $A \in \mathcal{O}(\Omega,\mathcal{H})$ defined on a three dimensional Hilbert space $\mathcal{H}$ 
as follows

\begin{eqnarray}
A_0 &= \frac{1}{2}(P_0 + P_1)  \nonumber \\
A_1 &= \frac{1}{2}(P_1 + P_2)   \\
A_2 &= \frac{1}{2}(P_2 + P_0)  \nonumber  \; 
\end{eqnarray}
where $P_i=\ket{i}\bra{i},i=0,1,2$ are projectors onto the orthonormal basis states $\ket{i}$. The Total Instrument we want to realize will be the Lüders instrument of $A$, whose action on a state $\rho$ is given by $\mathcal{T}_i(\rho)=\sqrt{A_i}\rho\sqrt{A_i}\;\;i=0,1,2$. We want to realize $\mathcal{T}$ as 
$2$-step ASI 
by applying Theorem \ref{main theorem}. 
Let us post-process the POVM $A$ into another POVM $B \in \mathcal{O}(\Omega_B,\mathcal{H})$ by taking 
\begin{align}
\nu =
\begin{pmatrix}
1 & 0  \\
1 & 0 \\
0 & 1
\end{pmatrix}
\quad\;\;
    \begin{aligned}
    \sqrt{B_0}&= \frac{1}{\sqrt{2}}P_0 + P_1 + \frac{1}{\sqrt{2}}P_2  \nonumber \\
    \sqrt{B_1}&= \frac{1}{\sqrt{2}}P_0 + \frac{1}{\sqrt{2}}P_2             
    \end{aligned}
\end{align} 

where $\Omega_B=\{0,1\}$. The Initial Instrument is $\mathcal{J}_j(\rho)=\sqrt{B_j} \rho \sqrt{B_j}; j=0,1$. 
Very simple calculation using Eq. (\ref{eq:RES1}) gives the following form of
the Residual Instruments  
\begin{align}  
     {\RI^{0}_0}(\rho)&= (P_0 + \frac{1}{\sqrt{2}}P_1)\,\rho\,(P_0 + \frac{1}{\sqrt{2}}P_1) \nonumber \\
    {{\RI}^{0}_1}(\rho)&= (\frac{1}{\sqrt{2}}P_1 + P_2)\,\rho\, (\frac{1}{\sqrt{2}}P_1 + P_2)    \\ 
    {{\RI}^{0}_2}(\rho)&=0  \nonumber \\
    {{\RI}^{1}_0}(\rho)&={\RI}^{1}_1(\rho)=0 \nonumber \\
 {{\RI}^{1}_2}(\rho)&= \rho \nonumber,
\end{align}
where choice of $c_{k,0m}$ is irrelevant due to $(I-\Pi_0)=0$ and we chose $c_{0,11}=c_{1,11}=0$, $c_{2,11}=1$. 
It is straightforward to check that the above equations define valid quantum instruments. 
In accordance with Remark \ref{re_outcomes} we see that 
$\RI^0_2 \circ \mathcal{J}_0=0$, $\RI^1_0 \circ \mathcal{J}_1=0$, $\RI^1_1 \circ \mathcal{J}_1=0$ proving that, if initial instrument $\mathcal{J}$ produced outcome zero, then Residual instrument $\RI^0$ can produce only outcomes $0$ or $1$ and similarly if initial instrument $\mathcal{J}$ produced outcome $1$, then Residual instrument $\RI^1$ can produce only outcome $2$.

\subsection{Resource implications of Theorem \ref{main theorem}} \label{subsec:Resource implications subsection}
Let us analyze the necessary ancilla size needed for sequential implementation of instruments proposed in Theorem \ref{main theorem}. 
Let's first concentrate on the case of non-shrinking Hilbert space ($\dim \mathcal{H} \leq \dim \mathcal{K}$). 
Let us denote by
\begin{align}
\label{def:gincr}
    g\equiv \lceil\dim \mathcal{K}/\dim\mathcal{H}\rceil 
\end{align}
the minimal dimension of additional quantum system with Hilbert space $\mathcal{V}$ that needs to be 
added 
to allow embedding of Hilbert space $\mathcal{K}$ in $\mathcal{H}\otimes\mathcal{V}$. If $\dim \mathcal{K}=\dim\mathcal{H}$ then $g=1$.
\tcb{From Corollary \ref{cor:KrausRanks} we have that} the number of non-zero Kraus operators of Residual instrument $\RI^j$ can be chosen to be 
\begin{align}
\label{eq:defmj}
    m_j=\sum_{k:\,\nu_{kj}>0} r_k,     
\end{align}
where $r_k$ is the Kraus rank of $\mathcal{T}_k$. If we assume the same ancillary system with dimension $d_A$ can be used to realize instrument $\mathcal{J}$ and after reset any of the instruments $\mathcal{R}^j , j\in \Omega_B$ then 
it implies the following statements. 
The ancilla dimension $d_A$ must be clearly at least the maximum of the needed dimension for the realization of the Initial instrument and any of the Residual instruments. 
As the initial instrument is always the Lüders instrument (i.e. requiring $d_A\geq |\Omega_B|$), the increase in output dimension 
happens in residual instruments $\mathcal{R}^j$. Consequently, for their realization some subsystem of ancillary system together with Hilbert space $\mathcal{H}$ must be used to hold the output state and the whole ancillary system must be factorizable into two parts with dimensions $g$ and $M$, 
where $M\geq m\equiv \max_{j\in \Omega_B}  m_j$. In turn, this implies that if $\max\{|\Omega_B|,gm\}=|\Omega_B|$
we might need to increase the dimension of the ancillary space to make it factorizable during the realization of the residual instruments $\mathcal{R}^j$. In conclusion,
ancilla must have dimension at least 
\begin{align}
\label{eq:dagen}
  d_A=\max(g \lceil \frac{|\Omega_B|}{g} \rceil, 
 \{g\, m_j\}_{j\in \Omega_B}),  
\end{align}
where the term $g$ accounts for the need of additional quantum systems needed to store the output on the larger output space and the fact that $d_A$ needs to be divisible by number $g$.

If there are several non-zero entries in the postprocessing matrix $\nu_{kj}$ for the fixed $k$ and some values of $j$, it implies that the rank $r_k$ is summed up in all such $m_j$'s (see Eq.(\ref{eq:defmj})). In other words, fractions of quantum operation $\mathcal{T}_k$ 
will be realized in several Residual instruments $\RI^j$, which increases the needed Kraus rank by $r_k$ for each of them. Thus, we see that having $\nu_{kj}$ equal to zero or one reduces the values of ranks $m_j$. In other words, coarse-grainings $A^{\mathcal{T}}\rightarrow B$ in general require fewer resources when Theorem \ref{main theorem} is applied. 
We tried to graphically illustrate these ideas by Figures \ref{FIG_1} and \ref{FIG_2}. The statement we can make is in a sense qualitative rather then quantitative. More precisely, consider applications of Theorem \ref{main theorem} for differently chosen POVM $B$. Among them, we claim that the subset for which the used ancilla size is minimal,  will always contain an instance in which POVM $B$ was created by coarse-graining of the induced POVM of the total instrument $A^\mathcal{T}$. 
Intuitively, every quantum operation $\mathcal{T}_k$ of the total instrument needs to be realized through some sequences of outcomes $(j,k)$ of the $2$-step 
ASI. 
Every such path implies that the instrument $\mathcal{R}^j$ needs to realize quantum operation $\mathcal{R}^j_k$ with the same Kraus rank as $\mathcal{T}_k$. Thus, every $\mathcal{R}^j$ needs $g\,r_k$ dimensions of the ancilla just to realize the final outcome $k$. Thus, minimizing the number of paths reaching every final outcome $k$ clearly minimizes the required ancilla size for instruments $\mathcal{R}^j$. On the other hand, if some of the instruments $\mathcal{R}^j$ do not fully utilize the ancilla for their implementation 
($g\, m_j < d_A$) 
adding a "path" $(j,k)$ through them might be possible. This could be done by properly modifying the postprocessing matrix $\nu_{kj}$ defining the POVM $B$ in Theorem \ref{main theorem}. While such a change will not affect (if we just lower, not nullify elements of $\nu_{kj}$) the required ancilla dimension for other residual instruments, it might however affect the form of the instrument that needs to be implemented, hence affecting the number of quantum gates needed for it's implementation. Thus, sequential implementations of instruments based on coarse-graining of outcomes does not need to be most effective from the gate count point of view. 

\begin{figure}[t]
\centering
  \includegraphics[width=8.2cm]{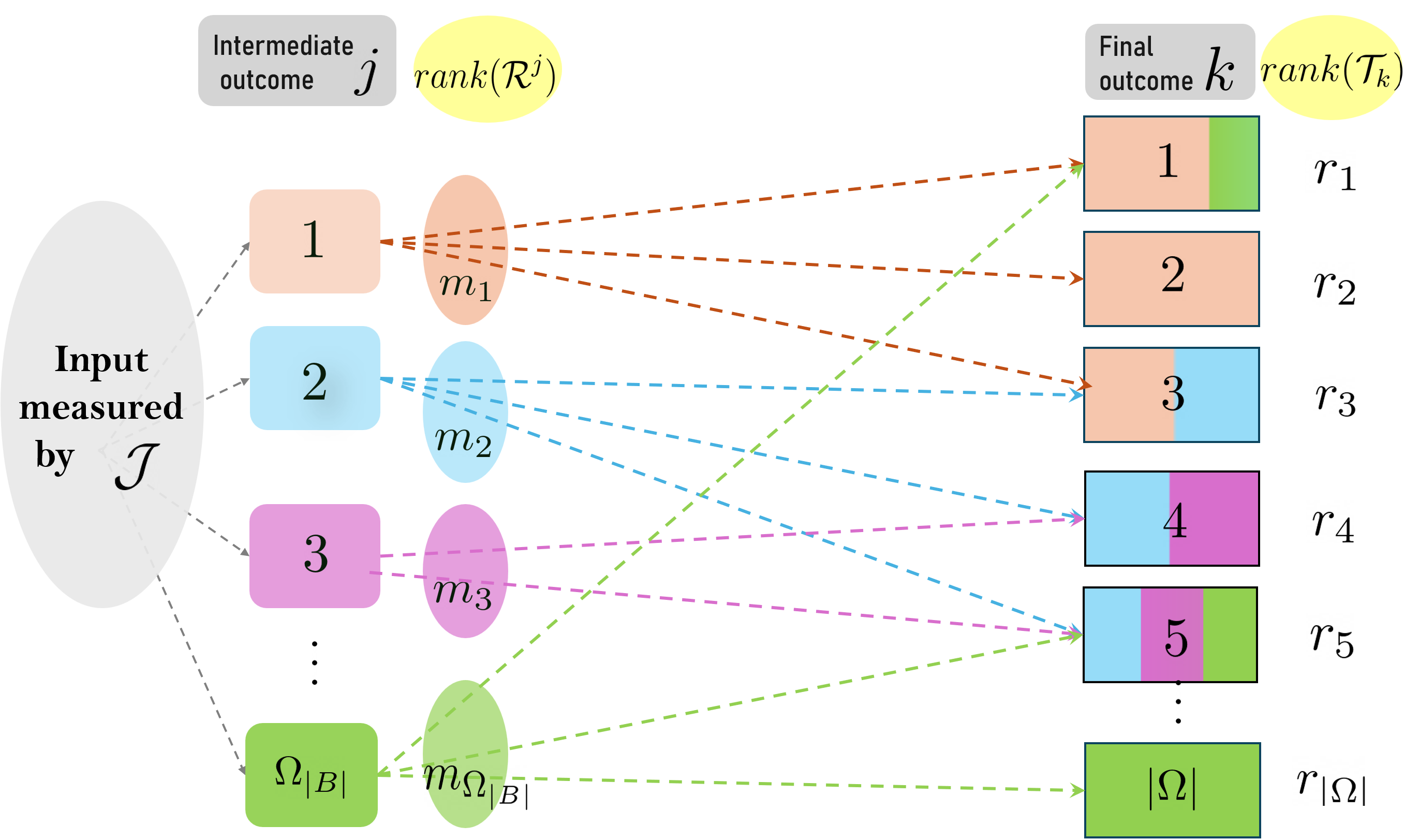}
  \captionof{figure}{Illustration of how postprocessing matrix $\nu_{kj}$ with several nonzero entries in some of the rows $k$ influences the Kraus rank of instruments $\RI^j$ and consequently the dimension of the ancilla system $d_A$.}
  \label{FIG_1}
 \end{figure}
 \begin{figure}
     \centering
  \includegraphics[width=8.2cm]{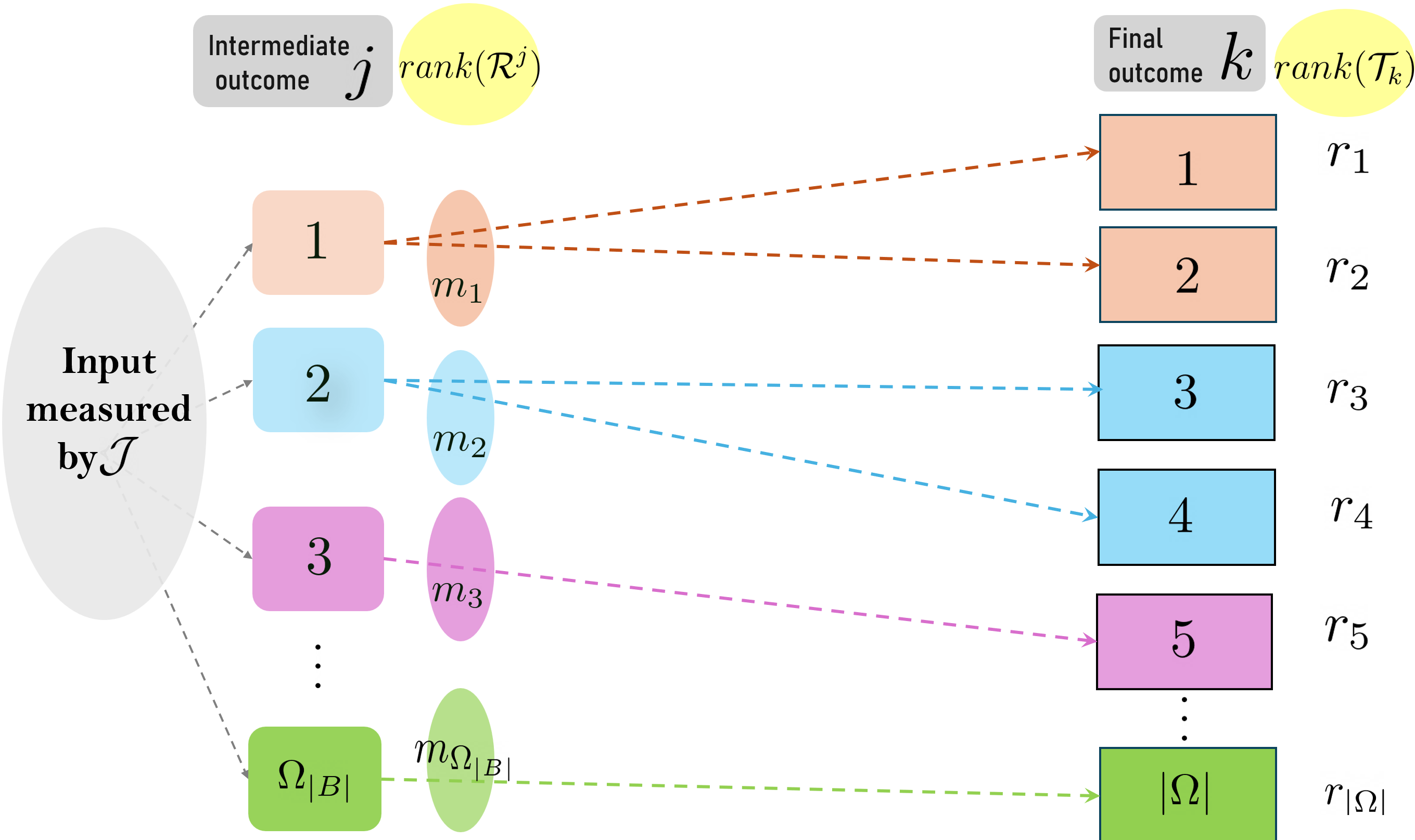} 
  \captionof{figure}{In contrast to Figure \ref{FIG_1}, if postprocessing matrix $\nu_{kj}$ contains only zeros or ones Kraus ranks of instruments $\RI^j$ are generally smaller, since each $r_k$ contributes to ranks of instruments $\RI^j$ just once.}
  \label{FIG_2}
\end{figure}

Above we have seen that to minimize the ancilla dimension it is meaningful to consider how to split outcomes $k$ of the Total instrument $\mathcal{T}$ into suitable number of disjoint subsets $\omega_j\subset \Omega$, where $j\in\Omega_B$, so that 
\begin{align}
    d_A=g\, \max(\lceil \frac{|\Omega_B|}{g} \rceil, \{ \sum_{k\in\omega_j} r_k\}_{j\in \Omega_B})
\end{align}
is minimized. This is in general quite complicated combinatorial task and it would be more pleasant to face the situation where each $r_k$ is equal to one. Indeed, this can be done by the following trick. For each quantum instrument, $\mathcal{T}$, one can define many detailed quantum instruments \cite{lm1}. These are instruments, which we obtain by taking some Kraus representation for each quantum operation of instrument $\mathcal{T}$ and we define a single outcome of the new quantum instrument by each of all those Kraus operators. 

For the purposes of this paper it will be useful to consider the following subset of those instruments.
\begin{definition}
Suppose we choose some minimal Kraus representations $\mathcal{T}_k(\rho)=\sum_{m=1}^{r_k} {T}_{k,m}\rho\,{T}_{k,m}^{\dag}$ for quantum operations $\mathcal{T}_k$ of quantum instrument $\mathcal{T}\in Ins(\Omega,\mathcal{H},\mathcal{K})$. 
For every such choice we 
define set $\Omega_D=\{(k,m)|k\in\Omega,m=1,\ldots,r_k\}$ as the set of all tuples $(k,m)$ 
and consequently 
we also define \emph{minimal detailed quantum instrument} $\mathcal{T}^D\in Ins(\Omega_D,\mathcal{H},\mathcal{K})$ by it's Kraus rank one quantum operations 
$\mathcal{T}_{k,m}(\rho)={T}_{k,m}\rho\,{T}_{k,m}^{\dag}$. 
\end{definition}

Obviously, by coarse-graining over the second element of $2$-tuple outcome $(k,m)$ of $\mathcal{T}^D$ we recover the Total instrument $\mathcal{T}$. Thus, if we implemented the quantum instrument $\mathcal{T}^D$ on real quantum hardware and, from the obtained outcome $(k,m)$, reported only the value of $k$, the addition of such simple classical postprocessing of outcomes would turn the experiment into an implementation of the instrument $\mathcal{T}$. An analogous idea for (sequential) implementation of POVMs was already used in \cite{MCMexpt2024}. In theory this approach works perfectly; in reality it might give rise to problems related to small statistics and errors introduced by imprecise measurements, since detailed instruments require more qubits to be measured. Such disadvantages are hard to estimate, as they heavily depend on the actual hardware. For the purposes of our further considerations, the main advantage of detailed instruments is that they have Kraus rank one ($r_k=1$) for all $k$. 
Such quantum instruments are also termed indecomposable instruments \cite{lm1}.

Thus, assuming availability of classical postprocessing of outcomes, we can without loss of generality consider finding an ancilla efficient implementation of instrument $\mathcal{T}^D$ instead of implementation for an instrument $\mathcal{T}$. In this case $\sum_{k\in\omega_j} r_k=|\omega_j|$, 
\begin{align}
\label{eq:sumrt}
    \sum_{j\in \Omega_B} |\omega_j| = \sum_{k\in\Omega_D} 1= |\Omega_D|=r_{\mathcal{T}}
\end{align}
and we want to minimize 
\begin{align}
\label{eq:minAnc}
    d_A=g\, \max(\lceil \frac{|\Omega_B|}{g} \rceil, \{|\omega_j|\}_{j\in \Omega_B}),
\end{align}
over choice of set $\Omega_B$ and cardinality of subsets $\omega_j \subset \Omega_D$ $j\in\Omega_B$. This optimization is most easily solved if we pose the question in the following way. What is the highest rank $r_{\mathcal{T}}$ of quantum instrument $\mathcal{T}_D$ that $2$-step ASI with ancilla size $d_A$ can realize. From the above considerations we anyway consider only ancilla sizes $d_A$ divisible by $g$. Obviously, the more elements set $\Omega_B$ has  the more positive summands in Eq. (\ref{eq:sumrt}) we can have. Thus, highest rank $r_{\mathcal{T}}$ with fixed $d_A$ is then achieved when 
$|\Omega_B|=d_A$. Realizability of instruments $\mathcal{R}^j$ with $d_A$ dimensional ancilla requires $g|\omega_j|\leq d_A$ $\forall j\in \Omega_B$, or equivalently Eq. (\ref{eq:minAnc}) must not be violated. The highest value each $|\omega_j|$ can have is an integer $d_A/g$. In such case we obtain the highest rank $r_{\mathcal{T}}$ 
that any $2$-step ASI with ancilla size $d_A$ can realize and it reads
\begin{align}
\label{eq:maxrt}
    r_{\mathcal{T}}=\sum_{j\in \Omega_B} |\omega_j| = d_A \,\frac{d_A}{g}=\frac{d_A^2}{g}
\end{align}
Thus, all instruments realizable (via Theorem \ref{main theorem}) with $d_A$ dimensional ancilla must obey $r_{\mathcal{T}}\leq\frac{d_A^2}{g}$. This can be seen as an lower bound on ancilla size $d_A$.    Taking into account that $d_A$ is divisible by $g$ we can write a bit more tight lower bound on $d_A$ as follows
\begin{align}
\label{eq:rtbound}
    \lceil\sqrt{\frac{r_{\mathcal{T}}}{g}}\rceil\leq\frac{d_A}{g}.
\end{align}
On the other hand, let's show that this bound is also tight, i.e. if we choose
ancilla size $d_A=g \lceil\,\sqrt{\frac{r_{\mathcal{T}}}{g}}\,\rceil$ we can realize quantum instrument $\mathcal{T}_D$ with rank $r_{\mathcal{T}}$ by a $2$-step ASI. 
To see that this is indeed possible we would choose $|\Omega_B|=d_A$. If all $|\omega_j|$ would be set to $\lceil\,\sqrt{\frac{r_{\mathcal{T}}}{g}}\,\rceil$ then we could realize an instrument with rank $r=\sum_{j\in \Omega_B} |\omega_j| = g \lceil\,\sqrt{\frac{r_{\mathcal{T}}}{g}}\,\rceil\,\lceil\,\sqrt{\frac{r_{\mathcal{T}}}{g}}\,\rceil\geq r_{\mathcal{T}}$. By suitably lowering values of $|\omega_j|$ (some can even be set to zero, which is equivalent to lowering $|\Omega_B|$) we can make $r=\sum_{j\in \Omega_B} |\omega_j|=r_{\mathcal{T}}$.

The above reasoning is summarized in the following theorem. 

\begin{theorem}
\label{th:opt2stepsplit}
Consider 
quantum instrument $\mathcal{T}\in Ins(\Omega,\mathcal{H},\mathcal{K})$
with Kraus rank 
$r_{\mathcal{T}}$ and such that $\dim \mathcal{H} \leq \dim \mathcal{K}$. 
Minimal ancilla dimension for implementation of quantum instrument $\mathcal{T}$ 
using $2$-step ASI \tcb{originating from Theorem \ref{main theorem}} and classical postprocessing of outcomes is given by 
\begin{align}
    d_A=g \lceil\,\sqrt{\frac{r_{\mathcal{T}}}{g}}\,\rceil,
\end{align}
where $g=\lceil\dim \mathcal{K}/\dim\mathcal{H}\rceil$ captures the ratio between dimension of output and input Hilbert space of the instrument $\mathcal{T}$.
\end{theorem}

For quantum instruments whose Kraus rank is not more than the ratio of the dimensions between output and input Hilbert space we immediately obtain the following corollary of Theorem \ref{th:opt2stepsplit}.


\begin{corollary}
\label{cor:low_mem_Inst_2s}
Suppose Kraus rank $r_{\mathcal{T}}$ of instrument $\mathcal{T} \in Ins(\Omega,\mathcal{H},\mathcal{K})$ obeys
$r_{\mathcal{T}}\leq \lceil\dim \mathcal{K}/\dim\mathcal{H}\rceil$. Then it's implementation with a $2$-step adaptive sequence of instruments (plus final postprocessing of classical outcomes) does not require other ancillary systems than those needed to form the final output space. In particular one can choose ancilla with dimension $d_A=\lceil\dim \mathcal{K}/\dim\mathcal{H}\rceil$, perform with it Lüders instrument $\mathcal{I}^{B}$ with $B=A^{\mathcal{T}^D}$, reset the ancilla and reuse it to realize an isometry conditionally chosen based on the obtained outcome of instrument $\mathcal{I}^{B}$. Finally, the intermediate classical outcome needs to be 
coarse-grained back from $\Omega_D$ to 
$\Omega$.
\end{corollary}

The above corollary represents an immediate advantage of ASI with respect to direct implementation of this type of  instrument, which would require $r_{\mathcal{T}}\lceil\dim \mathcal{K}/\dim\mathcal{H}\rceil$ dimensional ancilla.

Theorem \ref{th:opt2stepsplit} can be also interpreted as saying that for instruments with the same input and output Hilbert space the graph of correlations between outcomes in successive steps of ASI (see Figure \ref{FIG_2}) is optimally an equally branching tree, which in effect minimizes the needed ancilla dimension.


\subsection{Improvement of Theorem \ref{main theorem}}
\label{sec:imprvThm}

Let us turn our attention to instruments for which $\dim \mathcal{H} > \dim \mathcal{K}$. In this situation the benefit of using Theorem \ref{main theorem} to sequentially implement instrument $\mathcal{T}$ might be partly compromised by the need to introduce additional Kraus operators in the Residual instruments. 
In such cases it is helpful to search for sequential decompositions in which the intermediate Hilbert space dimension is as small as possible. Thus, instead of having ASI with $\mathcal{H}_0=\mathcal{H}_1=\mathcal{H}$, $\mathcal{H}_2=\mathcal{K}$ as in Theorem \ref{main theorem}, we would like to have $\mathcal{H}_1\neq \mathcal{H}$ and $\dim \mathcal{H}_1 < \dim \mathcal{H}$. This is possible to achieve if $\max_{j\in \Omega_B} rank(B_j)<\dim \mathcal{H}$ by fine-tuning Theorem \ref{main theorem} as follows.

\begin{theorem} \label{theorem3}
Given a Total Instrument $\mathcal{T}$ $\in Ins(\Omega,\mathcal{H},\mathcal{K})$ and a POVM $B \in \mathcal{O}(\Omega_B,\mathcal{H})$ such that the induced POVM of $\mathcal{T}$ can be postprocessed to POVM $B$ (i.e. $A^{\mathcal{T}}\rightarrow B$) and  
$d_1=\max_{j\in \Omega_B} rank(B_j)<\dim \mathcal{H}$, 
then there exists an adaptive sequence of an Initial Instrument $\tilde{\mathcal{J}}$ $\in Ins(\Omega_B,\mathcal{H},\mathcal{H}_1)$ and a set of Residual Instruments $\tilde{\RI}^{j} \in Ins(\Omega,\mathcal{H}_1,\mathcal{K})$ $\forall j \in \Omega_B$ such that
\begin{equation} 
    \mathcal{T}_k= \sum_{j\in \Omega_B} \tilde{\RI}^{j}_k \circ \tilde{\mathcal{J}}_{j} \quad \quad\forall k\in \Omega
    \end{equation}  
where $\dim \mathcal{H}_1=d_1$ and the Initial Instrument 
$\tilde{\mathcal{J}}_{j}(\rho)=K_j\rho K_j^\dagger, \forall \rho \in \mathcal{S(H)}$ is an indecomposable instrument with an induced POVM $B$ (i.e. $B_j=K_j^\dagger K_j$).
\end{theorem}
 
\begin{proof}
Let us choose some orthonormal basis $\{\ket{j}\}_{j=1}^{\dim \mathcal{H}}$ of Hilbert space $\mathcal{H}$. We define Hilbert space $\mathcal{H}_1$ as subspace of $\mathcal{H}$ generated by basis vectors $\{\ket{j}\}_{j=1}^{\dim \mathcal{H}_1}$.
Consider spectral decompositions of operators $B_j$
\begin{align}
    B_j=\sum_{k=1}^{rank(B_j)} \lambda^j_k \ket{v^j_k} \bra{v^j_k},
\end{align}
where $\lambda^j_k>0$ and vectors $\ket{v^j_k}$ are orthonormal. 
\tcb{For each $j$ we complete vectors $\ket{v^j_k}$ into orthonormal basis on $\mathcal{H}$. Thus, also values of $k>rank(B_j)$ address valid vectors $\ket{v^j_k}$.}
We define operators 
\begin{align}
    M_j= \sum_{k=1}^{rank(B_j)} \ket{k} \bra{v^j_k},
\end{align}
which map from Hilbert space $\mathcal{H}$ to $\mathcal{H}_1$ and we note that $M^\dagger_j M_j = \Pi_j$ is the projector onto the support of operator $B_j$ as in Theorem \ref{main theorem}.
Let us define instrument $\tilde{\mathcal{J}}$ by choosing it's Kraus operators to be $K_j=M_j\sqrt{B_j}$. 
It is easy to check normalization of instrument $\tilde{\mathcal{J}}$, since 
$\sum_j K_j^\dagger K_j=\sum_j \sqrt{B_j}\,\Pi_j\sqrt{B_j}=I$. 
Next, we define isometric channels 
\tcb{
$\mathcal{N}_j \in Ch(\mathcal{H}_1,\mathcal{H})$ as $\mathcal{N}_j(\rho)=N_j \rho N_j^\dagger$, where $N_j= \sum_{k=1}^{\dim \mathcal{H}_1} \ket{v^j_k} \bra{k}$} and we apply Theorem \ref{main theorem} on instrument $\mathcal{T}$ and obtain residual instruments $\RI^{j}$. This allow us to define residual instruments $\tilde{\RI}^{j}$ as $\tilde{\RI}^{j}=\RI^{j} \circ\,\tcb{ \mathcal{N}_j}$. Finally, we use $\tcb{N_j\, M_j}\sqrt{B_j}=\sqrt{B_j}$ to observe that 
\begin{align}
    (\tilde{\RI}^{j}_k \circ \tilde{\mathcal{J}}_{j})(\rho)&= \RI^{j}(\tcb{N_j} M_j\sqrt{B_j} \rho \sqrt{B_j}M_j^\dagger \tcb{N_j^\dagger}) \nonumber\\
    &= (\RI^{j}_k \circ \mathcal{J}_{j})(\rho)\quad \forall \rho\in \mathcal{S}(\mathcal{H}).
\end{align}
This concludes the proof, since due to Theorem \ref{main theorem}  $\forall k\in \Omega$
$\mathcal{T}_k= \sum_{j\in \Omega_B} \RI^{j}_k \circ \mathcal{J}_{j} = \sum_{j\in \Omega_B} \tilde{\RI}^{j}_k \circ \tilde{\mathcal{J}}_{j}$.
\end{proof}

Ideally, one would like to determine the minimal ancilla dimension also for the application of Theorem \ref{theorem3} to instruments that map from larger to smaller Hilbert spaces ($\dim \mathcal{H} > \dim \mathcal{K}$). However, the situation is rather complex, since there is no simple formula that relates the chosen postprocessing matrix to the ranks of the POVM elements $B_j$. In addition, also the number of Kraus operators needed for Residual instruments is not directly determined by the Total instrument, but depends on the dimensionality of the images of its Kraus operators as well as on the dimension of the intermediate Hilbert space $\mathcal{H}_1$. For this reason, we instead present an example illustrating the usefulness of Theorem \ref{theorem3} and leave the investigation of the resource implications of Theorem \ref{theorem3} for future work.

Let us consider the following example of a Total instrument $\mathcal{\tilde{T}}_i(\rho)=\tilde{K}_i\rho\tilde{K}^\dagger_i$,  $\tilde{\mathcal{T}}\in Ins(\{1,2,3\},\mathcal{H}_2^{\otimes 2},\mathcal{H}_2)$ mapping from $4$-dimensional Hilbert space to a $2$-dimensional Hilbert space, where

\begin{align}
\label{eq:defbadex}
    \tilde{K}_1&=\ket{0}\bra{00}+\frac{1}{\sqrt{2}}\ket{1}\bra{01} \nonumber\\
    \tilde{K}_2&=\frac{1}{\sqrt{2}}\ket{0}\bra{01}+\frac{1}{\sqrt{2}}\ket{1}\bra{10} \\
    \tilde{K}_3&=\frac{1}{\sqrt{2}}\ket{0}\bra{10}+\ket{1}\bra{11}. \nonumber
\end{align}

Direct single step implementation of this instrument would clearly require three dimensional ancilla. 
As we have seen in the discussion below Eq. (\ref{eq:dagen}) to explore resource saving applications of Theorem \ref{main theorem} it is sufficient to consider coarse-grainings of the three outcome induced POVM $A^{\tilde{\mathcal{T}}}$. There are essentially only three nontrivial options: 
i) join outcomes $1$ and $2$, ii) join outcomes $1$ and $3$, iii) join outcomes $2$ and $3$. 
On purpose the example is built, so that all these options are equivalent from the point of view of the relevant parameters, which are the rank of operators $\tilde{K}_i$, and the rank of potential POVM elements $B_j$. Thus, it is enough to discuss option i) to see resource implications of Theorem \ref{main theorem} and improvement if Theorem \ref{theorem3} is used instead. To see the difference, we define a two outcome POVM $B$ as
\begin{align}
    B_0&=\tilde{K}^\dagger_1 \tilde{K}_1 + \tilde{K}^\dagger_2 \tilde{K}_2    
    = I-B_1 \nonumber \\
    B_1&=\tilde{K}^\dagger_3 \tilde{K}_3 = \frac{1}{2}\ket{10}\bra{10}+\ket{11}\bra{11}.
\end{align}
It is easy to observe that $rank (B_0)=3$, $rank (B_1)=2$ and 
$\forall i \;\dim Im(K_i)=2$. 
Referring to the notation below the proof of Theorem \ref{main theorem}, this implies $d^0_{1,0}=d^0_{2,0}=2$, which means that if Theorem \ref{main theorem} is applied then $n_{add}=1$ additional Kraus operators are needed (see Eq. \ref{eq:addKraus}) for both Residual instrument $\mathcal{R}^0$ and $\mathcal{R}^1$. This means instruments $\mathcal{R}^0$ and $\mathcal{R}^1$ have Kraus rank $3$ and $2$, respectively. Therefore, corresponding 2-step ASI requires at least three dimensional ancilla to realize it's instruments. 
Therefore, usage of Theorem \ref{main theorem} provides no advantage with respect to direct dilation of the Total instrument $\tilde{\mathcal{T}}$. 
On the other hand, if Theorem \ref{theorem3} is used instead, one can choose a three-dimensional Hilbert space $\mathcal{H}_1$. Thanks to this, there is no need for additional Kraus operator for $\mathcal{R}^0$ and again one additional Kraus is needed for $\mathcal{R}^1$. However, in this case both $\mathcal{R}^0$ and $\mathcal{R}^1$ have Kraus rank $2$ and a two-dimensional ancilla is sufficient for realization of the corresponding ASI. Thus, we demonstrated an example in which Theorem \ref{theorem3} provides resource saving ASI in comparison to both direct dilation of the instrument and use of Theorem \ref{main theorem}.

\section{
Implementation of an Instrument via $N$-step ASI}
\label{sec:n-asi}
The aim of this section is to use Theorem \ref{main theorem} recursively to show that one can decompose an instrument also into an $N$-step ASI.
Results presented below fully 
rely on the notion of an adaptive sequence of instruments (ASI), which was  introduced in Section \ref{sec1} and we believe that it provides precise and readable notation in the following proofs and discussions.
\begin{corollary} 
\label{corollary 1}
    Given a Total instrument $\mathcal{T}\in Ins(\Omega_N,\mathcal{H}_0,\mathcal{H}_N)$ and the postprocessing relations $A^{\mathcal{T}}\rightarrow B^{N-1}\rightarrow \cdots \rightarrow B^{1}$ for 
    POVMs $B^{k} \in \mathcal{O}(\Omega_k,\mathcal{H}_0)$ 
    $k=1,\cdots,N-1$, there exists an adaptive sequence of instruments $\mathcal{Q}=({\mathcal{Q}^1,\cdots, \mathcal{Q}^N})$ in which the instrument applied in the $k^{th}$ step is denoted by $\mathcal{I}^{k,a_{k-1}}\in \mathcal{Q}^k$, the Initial instrument $\mathcal{I}^{1,1}$ is the Lüders instrument of POVM $B^{1}$ and 
    \begin{equation}
      \label{eq:splittoN}
     \mathcal{T}_{a_N}=\!\!\sum_{a_{N-1} \in \Omega_{N-1}}\!\!\!\cdots\sum_{a_1 \in \Omega_1} \mathcal{I}^{N,a_{N-1}}_{a_N}\!\circ \ldots \circ\mathcal{I}^{2,a_1}_{a_2}\circ\mathcal{I}^{1,1}_{a_1}.
     \end{equation}    
\end{corollary}

\begin{proof}
We will apply the Theorem (\ref{main theorem}) recursively to prove our result.  
In order to clearly distinguish different uses of Theorem \ref{main theorem}, we will add a fixed superscript to all involved objects.
At first, we will use the postprocessing relation $A^{\mathcal{T}}\rightarrow B^{N-1}$ and decompose the Total Instrument $\mathcal{T}\equiv\mathcal{T}^N$ using Theorem \ref{main theorem} into a two-step adaptive sequence. 
We use the obtained Residual Instruments to define the last step of the $N$ step adaptive sequence. Thus, 
$\mathcal{Q}^N=\{\mathcal{I}^{N,a_{N-1}}\equiv\RI^{N,a_{N-1}}\}_{a_{N-1}\in \Omega_{N-1}}$. We denote the obtained initial instrument as $\mathcal{J}^{N}$.
Next, we use the Initial Instrument $\mathcal{J}^{N}$ as the Total Instrument in Theorem \ref{main theorem}, i.e. $\mathcal{T}^{N-1}=\mathcal{J}^{N}$. Thus, at this step $\mathcal{T}^{N-1}\in Ins(\Omega_{N-1},\mathcal{H}_0,\mathcal{H}_0)$ and $A^{{\mathcal{T}}^{N-1}} = B^{N-1}$. Using the postprocessing relation $B^{N-1}\rightarrow B^{N-2}$ in Theorem \ref{main theorem}, we obtain a new set of residual instruments, which we use to define the one-to-last step of the adaptive sequence. 

$\mathcal{Q}^{N-1}=\{\mathcal{I}^{N-1,a_{N-2}}\equiv\RI^{N-1,a_{N-2}}\}_{a_{N-2}\in \Omega_{N-2}}$. 

We remind that currently
\begin{align}
    \sum_{\substack{a_{N-1} \in \Omega_{N-1} \\ a_{N-2} \in \Omega_{N-2}}} 
    \mathcal{I}^{N,a_{N-1}}_{a_{N}} \circ\mathcal{I}^{N-1,a_{N-2}}_{a_{N-1}} \circ \mathcal{J}^{N-1}_{a_{N-2}}&= \nonumber \\ 
    \sum_{\substack{a_{N-1} \in \Omega_{N-1} \\ a_{N-2} \in \Omega_{N-2}}}     
    \mathcal{I}^{N,a_{N-1}}_{a_{N}} \circ\RI^{N-1,a_{N-2}}_{a_{N-1}} \circ \mathcal{J}^{N-1}_{a_{N-2}}&= \nonumber \\
    \sum_{a_{N-1} \in \Omega_{N-1}} \mathcal{I}^{N,a_{N-1}}_{a_{N}} \circ\mathcal{T}^{N-1}_{a_{N-1}}&= \nonumber \\
    \sum_{a_{N-1} \in \Omega_{N-1}} \RI^{N,a_{N-1}}_{a_{N}} \circ\mathcal{J}^{N}_{a_{N-1}}&= \nonumber \\        
=\mathcal{T}^N_{a_N}
\end{align}

By performing the described recursion process $(N-1)$ times in total, we determine the whole adaptive sequence of instruments if, in the last recursion step, we fix $\mathcal{Q}^{1}=\{\mathcal{I}^{1,1}\equiv\mathcal{J}^{2}\}$. We note that all instruments in the sets $\mathcal{Q}^1,\cdots, \mathcal{Q}^{N-1}$ are indecomposable, since they originate from application of Theorem \ref{main theorem} onto a 
Lüders instrument.
\end{proof}


We note that the resource consideration preceding Theorem \ref{th:opt2stepsplit} can be generalized to the $N$-step scenario. Assuming that $\dim \mathcal{H} \leq \dim \mathcal{K}$ and that we are applying Corollary \ref{corollary 1} to the detailed instrument $\mathcal{T}^D$, the highest Kraus rank of $\mathcal{T}^D$, which any $N$-step ASI can implement is given by 
\begin{align}
    \label{eq:nstepbound}
    r_{\mathcal{T}}\leq d_A^{N-1}\frac{d_A}{g},
\end{align}
where $g$ is given by Eq. (\ref{def:gincr}). This implies attainable lower bound for ancilla dimension, i.e. the minimal ancilla dimension given by
\begin{align}
\label{eq:dA_gen}
     d_A=g \lceil\,\sqrt[N]{\frac{r_{\mathcal{T}}}{g^{N-1}}}\,\rceil
     = g \lceil\,\frac{1}{g}\sqrt[N]{g\, r_{\mathcal{T}}}\;\rceil.
\end{align}

In situations when the number of qubits for the Hilbert space $\mathcal{H}$ and the size of ancilla usable for quantum computations is fixed by the available hardware, Eq. (\ref{eq:nstepbound}) can be rewritten as a lower bound on the number of steps of any ASI (originating from Corollary \ref{corollary 1})  capable to implement the chosen detailed 
instrument $\mathcal{T}^D$. In particular, we get
\begin{align}
    \label{eq:boundforN}
    N\geq \log_{d_A}{(g\,r_{\mathcal{T}})}.
\end{align}
Another interesting way to rewrite Eq. (\ref{eq:nstepbound}) is to put on one side the quantities fixed by the choice of the Total instrument and on the other the quantities we can trade-off, i.e. the number of steps $N$ of the ASI and the dimension of the used ancillary system $d_A$ or equivalently the corresponding number of ancillary qubits $d_A=2^{n_A}$. Thus, we obtain $g\,r_{\mathcal{T}}\leq 2^{(n_A\,N)}$ or equivalently
the following corollary.

\begin{corollary}
\label{cor:tradeoff1}
    Given a Total instrument $\mathcal{T}\in Ins(\Omega_N,\mathcal{H}_0,\mathcal{H}_N)$
    with $\dim\mathcal{H}_0\leq\dim \mathcal{H}_N$
    any $N$-step adaptive sequence of instruments, which implements it using Corollary \ref{corollary 1}, obeys the inequality 
\begin{align}
    \label{eq:tradeoff1}
    N\,n_A \geq \log_2{(g\, r_{\mathcal{T}})}
\end{align}
that requires
the number of steps of the considered ASI times the number of used ancilla qubits $n_A$ to be at least the logarithm of the product of the dimension increase factor $g=\lceil\dim \mathcal{H}_N/\dim\mathcal{H}_0\rceil$ and the rank of Total instrument.  
\end{corollary}

\subsection{General limitations of N-step ASI}
\label{sec:limitations_N-ASI}
Considerations from the previous subsection can be slightly generalized to characterize implementation possibilities of an arbitrary $N$-step ASI, not just those produced via Corollary \ref{corollary 1}. As in the rest of the paper, we want to consider the demand for quantum system size in the same way. Thus, if in the first step a $d_A \equiv d_{A_0}$ dimensional ancilla is used together with the main quantum system of dimension $d_0=\dim {\mathcal{H}_0}$ then after the first step of ASI, which ends with a $d_1$ dimensional system only $d_{A_1}=\lfloor \frac{d_0 d_A}{d_1} \rfloor$ dimensional ancillary system can be used for the implementation of the second step of the ASI. Analogously, after the $k$-th step a $d_{A_k}=\lfloor \frac{d_0 d_A}{d_k} \rfloor$ dimensional ancilla is available.

Thus, let us assume we have an arbitrary $N$-step ASI, which uses $d_A$dimensional ancilla and implements Total instrument $\mathcal{T} \in Ins(\Omega_N,\mathcal{H}_0,\mathcal{H}_N)$. Our goal is to find an upper bound on the rank of $\mathcal{T}$.

Let us start with just a $2$-step ASI. Due to Eq. (\ref{eq:splittoN}) we have that
\begin{align}
    r_{a_2}\leq \sum_{a_1} r^{2,a_1}_{a_2} r^{1,1}_{a_1},
\end{align}
where we denoted by $r_{a_2}, r^{2,a_1}_{a_2}, r^{1,1}_{a_1}$ the minimal Kraus rank of $\mathcal{T}_{a_2}, \mathcal{I}^{2,a_1}_{a_2}, \mathcal{I}^{1,1}_{a_1}$, respectively and we used how minimal Kraus rank behaves under sum and composition. Since $r_{\mathcal{T}}=\sum_{a_2} r_{a_2}$ we obtain
\begin{align}
    \label{eq:rankbound}
    r_{\mathcal{T}}\leq \sum_{a_1}(r^{1,1}_{a_1}\sum_{a_2} r^{2,a_1}_{a_2}).
\end{align}
Maximum rank of instrument $\mathcal{I}^{k,a_{k-1}}\in  Ins(\Omega_k, \mathcal{H}_{k-1},\mathcal{H}_k)$, which can be implemented with $d_{A_{k-1}}$ dimensional ancilla is $d_{A_k}$, because that is the dimensionality of the system on which we can perform projective measurement (see also discussion in Section \ref{sec1}). Thus, we have
\begin{align}
    \sum_{a_k} r^{k,a_{k-1}}_{a_k}\leq d_{A_k}.
\end{align}
Inserting this into Eq. (\ref{eq:rankbound}) we obtain
\begin{align}
    \label{eq:rankbound2s}
    r_{\mathcal{T}}\leq \sum_{a_1}r^{1,1}_{a_1}d_{A_2}\leq d_{A_1}d_{A_2}=\lfloor\frac{d_0 d_A}{d_1}\rfloor \;\lfloor\frac{d_0 d_A}{d_2}\rfloor.
\end{align}
Repeating the same steps and logic for $N$-step ASI gives
\begin{align}
    \label{eq:rankboundNs}
    r_{\mathcal{T}}\leq d_{A_1}\ldots d_{A_N}= \lfloor\frac{d_0 d_A}{d_1}\rfloor \dots \lfloor\frac{d_0 d_A}{d_N}\rfloor.
\end{align}
If for a fixed number of steps $N$ we want to achieve the highest possible rank $r_{\mathcal{T}}$ then we need to lower dimensionality of intermediate Hilbert spaces $\mathcal{H}_k$, $k=1,\ldots,N-1$ as much as possible. Corollary \ref{corollary 1} shows that $d_k=d_0$ for $k=1,\ldots,N-1$ is always possible leading to 
\begin{align}
    \label{eq:rankboundNsd0}
    r_{\mathcal{T}}\leq d_A^{N-1} \lfloor\frac{d_0 d_A}{d_N}\rfloor.
\end{align}
For $\dim \mathcal{H} \leq \dim \mathcal{K}$ inequality (\ref{eq:rankboundNsd0}) (almost precisely) coincides with (\ref{eq:nstepbound})    
if we take into account that $d_A$ should be chosen to be divisible by $g$. If we consider for simplicity that all the involved dimensions are powers of two to avoid complications caused by indivisibility, since one typically works with systems of qubits, then the two inequalities match exactly.
Advantage of inequality (\ref{eq:rankboundNs}) is that it is valid also for the case of  $\dim \mathcal{H} > \dim \mathcal{K}$ and we saw in Theorem \ref{theorem3} that also $d_k<d_0$ is possible with the methods developed in this paper.

From inequality (\ref{eq:rankboundNs}) we see that if we cannot lower dimensions $d_k$ significantly below $d_0$ then  tradeoff between number of steps $N$ of a general ASI and the number of qubits $n_A$ of the used ancilla will essentially have the form 
\begin{align}
    \label{eq:gentradeoff}
    N n_A\geq const
\end{align}
similarly as in (\ref{eq:tradeoff1}). It seems quite likely that a given set of quantum operations does not ignore some subspace of the Hilbert space, so one is tempted to think that the tradeoff (\ref{eq:gentradeoff}) might be generic for general ASI implementing a fixed Total instrument.

\subsection{Instruments with N-tuple outcome space} \label{Instruments with n-tuple outcome space}

Let us consider a situation in which the outcome space of the considered Total instrument $\mathcal{T}$ is an $N-$fold Cartesian product space $\Omega = \lambda_1 \times\ldots \times\lambda_N$. In this subsection, we will show that instrument $\mathcal{T}$ can be realized using an 
$N$-step ASI 
$\mathcal{Q}$ such that its outcome 
$(a_1,\ldots,a_N)\in \Omega$ is produced by obtaining an element $a_k$ at the $k^{th}$ step of the adaptive sequence $\mathcal{Q}$.


We will do this by demonstrating that the considered scenario is a special case of 
Corollary \ref{corollary 1}. Thus, we define $\Omega_k = \lambda_1 \times \ldots\times \lambda_k$ and POVMs $B^k\in \mathcal{O}(\Omega_k,\mathcal{H})$ as 
\begin{align} 
\label{eq:def_marg_POVM}
   B^k_{a_1,\ldots,a_k}=\sum_{a_{k+1}\in \Omega_{k+1}} B^{k+1}_{a_1,\ldots,a_{k+1}}, 
\end{align}
for all $k=1,\ldots,N-1$ with $B^N\equiv A^{\mathcal{T}}$. 
Clearly, marginalization over the rightmost outcome, which relates POVMs $B^{k+1}$ and $B^{k}$, can be understood as the postprocessing relation $B^{k+1}\rightarrow B^k$ by setting 
\begin{align}
\label{eq:numarg}
    \nu_{\vec{a}_{k+1}\, \vec{a}'_{k}}=\delta_{a_1 a'_1}\ldots\delta_{a_k a'_k} \quad \forall \vec{a}_{k+1}\in\Omega_{k+1}\;\;\forall \vec{a}'_{k}\in\Omega_k,
\end{align}
where we used the notation $\vec{a}_{m}\equiv (a_1,\ldots,a_{m})$. 
Clearly, $\forall \vec{a}_{k+1}$ $\sum_{\vec{a}'_{k}\in\Omega_k} \nu_{\vec{a}_{k+1}\, \vec{a}'_{k}}=1$ and $\nu_{\vec{a}_{k+1}\, \vec{a}'_{k}}$ is nonzero only if $a_m=a'_m \; \forall m=1,\ldots,k$. Due to 
Corollary \ref{corollary 1}, the Total Instrument $\mathcal{T}$ can be realized as an ASI  
in which outcomes at the $k^{th}$ step belong to $\Omega_k$ and their probability of appearance depends on the POVM $B^k$. Thus, if outcome $a'_1$ is measured in the first step, then outcome $\vec{a}_{2}=(a_1,a_2)$ will be obtained in the second step. Here, we must obtain 
an outcome such that $\nu_{\vec{a}_{2}\, \vec{a}'_{1}}$ is nonzero (see Remark \ref{re_outcomes}). As we have mentioned above, this happens only if $a_1=a'_1$. 

Thus, due to Eq.(\ref{eq:numarg}) we can simplify 
Eq. (\ref{eq:splittoN}), 
which upon applying Corollary \ref{corollary 1} to the considered scenario appears as follows
\begin{align}
 \mathcal{T}_{\vec{a}_N}&=\!\!\sum_{\vec{a}'_{N-1} \in \Omega_{N-1}}\!\!\!\!\!\cdots\sum_{\vec{a}'_1 \in \Omega_1} \mathcal{I}^{N,\vec{a}'_{N-1}}_{\vec{a}_N}\!\circ \ldots \circ\mathcal{I}^{2,\vec{a}'_1}_{\vec{a}'_2}\circ\mathcal{I}^{1,1}_{\vec{a}'_1} \nonumber \\
 &=\mathcal{I}^{N,\vec{a}_{N-1}}_{\vec{a}_N}\!\circ \ldots \circ\mathcal{I}^{2,\vec{a}_1}_{\vec{a}_2}\circ\mathcal{I}^{1,1}_{\vec{a}_1}, 
\end{align}
where in the first line, $\vec{a}'_k$ stands for independent $k$ tuple variables, while in the second line, 
$\vec{a}_k$ appearing on the Right Hand Side of equation are fixed by $N$-tuple $\vec{a}_N$ from the Left Hand Side. In conclusion, the above findings can be summarized in the following corollary.


\begin{corollary} \label{corollary 2}
    Given a Total instrument $\mathcal{T}\in Ins(\lambda_1\times\ldots \times\lambda_N,\mathcal{H}_0,\mathcal{H}_N)$ there exists an adaptive sequence of instruments $\mathcal{Q}=({\mathcal{Q}^1,\cdots \mathcal{Q}^N})$ in which the instrument 
    $\mathcal{I}^{k,\vec{a}_{k-1}}\in \mathcal{Q}^k$ applied in the $k^{th}$ step determines the $k^{th}$ element of the overall outcome $\vec{a}_N=(a_1, \ldots, a_N)$ and 
\begin{align} \label{eq:n tuple long}
 \mathcal{T}_{\vec{a}_N} =
 \mathcal{I}^{N,\vec{a}_{N-1}}_{\vec{a}_N}\!\circ \ldots \circ\mathcal{I}^{2,\vec{a}_1}_{\vec{a}_2}\circ\mathcal{I}^{1,1}_{\vec{a}_1}, 
\end{align}
where the Initial instrument $\mathcal{I}^{1,1}$ is the Lüders instrument.
\end{corollary}



\begin{remark}
   If $\mathcal{T}\in Ins(\Omega_{2^N},\mathcal{H},\mathcal{H})$ is the Lüders instrument 
    having $|\Omega|=2^N$ outcomes, then the single step   
   instrument dilation requires at least $N$-qubit ancilla. By choosing any one to one mapping of $\Omega_{2^N}$ to an $N$-fold Cartesian product space $\Omega = \lambda_1 \times\ldots \times\lambda_N$ of $2$ outcome sets and using Corollary \ref{corollary 2}, one can obtain $N$-step ASI, which uses only a single qubit ancilla for realization of each of it's $N$ steps. 
\end{remark}

The situation above can be viewed as describing two extreme cases of a trade-off between time and space resources: having $N$ additional ancilla qubits used for a single dilation step, or having just a single ancilla that is repeatedly reused ancilla for $N$ consecutively performed dilation steps. The saving of such space resources comes at the price of requiring an adaptive quantum circuit and the ability to store classical information at each intermediate step. 
The benefits of one or the other implementation of the same instrument may depend heavily on the particular quantum hardware and available resources. We comment on this further in Section \ref{sec:summary}.

\subsection{Smallest ancilla size implementation of a Quantum instrument}

Let us now return to Eq. (\ref{eq:dA_gen}), which implies that the dimension of the ancilla is at least $d_A \geq g$. One can see that this lower bound can be achieved for sufficiently large $N$. This motivates us to formulate a generalization of Corollary \ref{cor:low_mem_Inst_2s} to the $N$-step context.

\begin{corollary}
\label{cor:low_mem_Inst_Ns}
For any instrument $\mathcal{T} \in Ins(\Omega,\mathcal{H},\mathcal{K})$ with $g= \lceil\dim \mathcal{K}/\dim\mathcal{H}\rceil > 1$ there exist $N$ such that it can be implemented with a $N$-step 
adaptive sequence of instruments (plus final postprocessing of classical outcomes), which does not require other ancillary systems than those needed to form the final output space ($d_A=g$). 
In particular, one can choose POVMs in Corollary \ref{corollary 1} so that 
after the $(N-1)$-th step of ASI ancilla is reset and reused to realize an isometry conditionally chosen based on the obtained outcome $a_{N-1}$. Finally, outcome $a_{N-1}$ needs to be 
coarse-grained back from $\Omega_D$ to the set 
$\Omega$.
\end{corollary}

\begin{proof}
   We begin by choosing lowest $N$ for which Eq. (\ref{eq:dA_gen}) returns $d_A=g$, so $r_{\mathcal{T}}\leq g^{N-1}$. Next, we define $\lambda_1=\ldots=\lambda_{N-1}=\{1,\ldots,g\}$, and set $\Omega_k = \lambda_1 \times \ldots\times \lambda_k$ for $k=1,\ldots,N-1$. We also define $\Omega_N=\Omega_{N-1}$. As a next step we 
   construct detailed instrument $\mathcal{T}^D$ and split its outcome set $\Omega_D$ ($|\Omega_D|\leq g^{N-1}$) into $g^{N-1}$ (nonintersecting) subsets 
   $\{\omega_j\}_{j\in \Omega_{N-1}}$ each containing zero or one element. We define POVMs $B^{N}=B^{N-1} \in \mathcal{O}(\Omega_{N-1},\mathcal{H}_0)$ via equation
   $B^{N-1}_j=\sum_{k\in \omega_j} A^{\mathcal{T}^D}_k$. Thus, we effectively added 
   zero POVM elements to formally have $g^{N-1}$ outcomes for POVMs $B^{N}=B^{N-1}$. We proceed by defining POVMs  $B^k\in \mathcal{O}(\Omega_k,\mathcal{H})$ $k=1,\ldots,N-2$ via marginalization over the rightmost outcome as in Eq. (\ref{eq:def_marg_POVM}).  
   Application of Corollary \ref{corollary 1} to the current setting has similar consequences as in Corollary \ref{corollary 2}. Most important for us is that for $k\in\{1,\ldots,N-1\}$ instruments $\mathcal{I}^{k,\vec{a}_{k-1}}\in Ins(\Omega_k,\mathcal{H}_0,\mathcal{H}_0)$ have at most $g$ nonzero quantum operations $\mathcal{I}^{k,\vec{a}_{k-1}}_{a_k}$, which follows from Remark \ref{re_outcomes} and postprocessing relation being given by Eq. (\ref{eq:numarg}). Thus, realization of instruments $\mathcal{I}^{k,\vec{a}_{k-1}}$ is possible with $g$-dimensional ancilla, which used for measurement and reset in the first $N-1$ steps of the ASI. In the 
   $N$-th step of the constructed ASI, due to $B^N=B^{N-1}$ the postprocessing matrix is an identity, so due to Remark \ref{re_outcomes} there is just a single possible outcome $\vec{a}_{k-1}$ of the residual instruments and so these instruments $\mathcal{I}^{N,\vec{a}_{N-1}}\in Ins(\Omega_N,\mathcal{H}_0,\mathcal{H}_N)$ must be isometries. We conclude that in order to realize them $g$ dimensional ancilla is sufficient, since no measurement is needed at $N$-th step of the constructed ASI.
\end{proof}


\subsection{Implementing qubit four outcome Lüders instruments with one ancilla qubit}

Let us consider a practical example of Corollary \ref{corollary 2} by considering 
a qubit four outcome Lüders instrument. Suppose we have a two-qubit quantum computer, hence direct $1$-step dilation, which requires one qubit for the system and two qubits for the ancilla cannot be performed. On the other hand, Corollary \ref{corollary 2} can be used to solve the problem at hand in the following way. We choose a relabeling of the four outcomes of the instrument to be implemented into $2$-tuples $(a_1,a_2)$ $a_1,a_2=0,1$ and we denote it $\mathcal{T}\in Ins(\Omega_A,\mathcal{H}_2,\mathcal{H}_2)$, where $\Omega_A=\{0,1\}\times \{0,1\}$ and $\mathcal{H}_2$ denotes two-dimensional Hilbert space. At this point, it suffices to combine Corollary \ref{corollary 2} with a dilation 
of the Lüders instrument to obtain the final result. Here we assume that the ancillary qubits are undergoing reset to some fixed state, before every step of the adaptive sequence. The above procedure works for any ranks of the POVM elements defining the decomposed instrument. 

In order to be less abstract, we provide below a case study for three parameter class of qubit four outcome Lüders instruments, which are defined (see Figure \ref{fig:4outcomePovm}) by the following 
POVM $A\in O(\Omega_{A},\mathcal{H}_2)$:
\begin{align}
\label{eq:defPovmA}
    A_{00}&=\frac{1}{4}[I +\eta(-\sigma_x\tan{\beta} -\sigma_z)]  \nonumber \\
     A_{01}&=\frac{1}{4}[I +\eta(\sigma_x\tan{\beta} -\sigma_z)]  \\
      A_{10}&=\frac{1}{4}[I +\eta(-\sigma_y\tan{\alpha} +\sigma_z)]  \nonumber\\
       A_{11}&=\frac{1}{4}[I +\eta(\sigma_y\tan{\alpha} +\sigma_z)], \nonumber
\end{align}
where $0< \alpha,\beta < \frac{\pi}{2}$, 
$0< \eta \leq \min{\{ \cos{\alpha},\cos{\beta}\}}$ 
and $\sigma_x,\sigma_y,\sigma_z$ are the Pauli matrices. 
To improve readability of the formulas in this example we shortened subscripts from $(j,k)$ to just $jk$. 
Without loss of generality, we assume that $\alpha\geq\beta$. POVM 
$A$ is an informationally complete POVM, and it coincides with symmetric informationally complete (SIC)-POVM  \cite{SIC} for 
$\alpha=\beta=\arccos({\frac{1}{\sqrt{3}}})$ and $\eta=\frac{1}{\sqrt{3}}$. 
The Total Instrument $\mathcal{T}\in Ins(\Omega_{A},\mathcal{H}_2,\mathcal{H}_2)$ acting on a state $\rho$ is given by 
 $\mathcal{T}_{jk}(\rho) =\sqrt{A_{jk}}\, \rho \,\sqrt{A_{jk}}$, 
which we can represent using Corollary \ref{corollary 2} as
\begin{equation} \label{eq:exampplleeeeeXX}
    \mathcal{T}_{jk} = \mathcal{I}^{2,{j}}_{jk}\circ \mathcal{I}^{1,1}_{j}
\end{equation}

\begin{figure} 
\includegraphics[height=8.5cm,width=7.5cm]{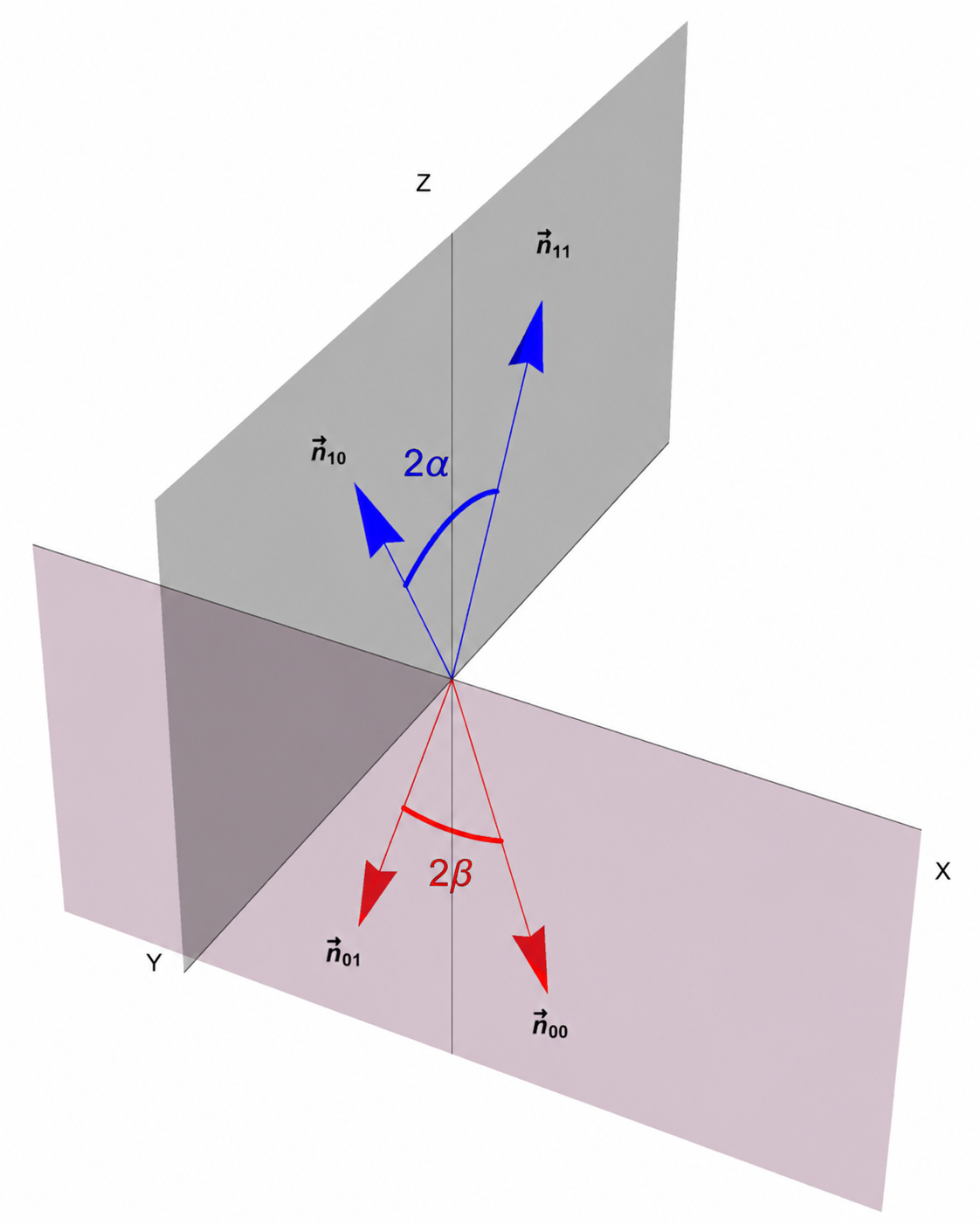}
\caption{
\label{fig:4outcomePovm}
Illustration of directions and parameters that define the POVM $A$. When $A$ is a symmetric informationally complete POVM then ends of the four vectors form a regular tetrahedron. Parameter $\eta$ (see Eq. (\ref{eq:defPovmA})) is not depicted and relates to the sharpness of the POVM 
elements. 
}
\end{figure}

Thus, outcome $j$ will be obtained in the first step of the ASI, 
and based on it, instrument $\mathcal{I}^{2,{j}}$ will be applied in the second step. In particular, instrument $\mathcal{I}^{2,{j}}$ determines outcome $k$, i.e. the second element of the two-tuple $(j,k)$. 
In accordance with Eq. (\ref{eq:def_marg_POVM}) we have
$B^{1}_{j}=\sum_{k=0}^{1} A_{jk}$ $j=0,1$. 
Since the input and output Hilbert space of the instrument $\mathcal{T}$ have the same dimension no additional Kraus operators will be needed. In particular, POVM elements of $B^1_{j}$ are full rank (sum of non-collinear of positive operators of rank one in $\mathcal{H}_2$ is of full rank) and consequently each of quantum operations of the instruments $\mathcal{I}^{1,1}$, $\mathcal{I}^{2,j}$, has just a single Kraus operator. We also remind that $(B_{j}^1)^{-\frac12}  (B_{j}^1)^{\frac12}=I$ and we denote those single Kraus operators of 
quantum operations $\mathcal{I}^{2,j}_{jk}$, $\mathcal{I}^{1,1}_{j}$ as $K^{2,j}_{jk}$ and $K^{1,1}_{j}$, respectively.
In conclusion, in accordance with Eq. (\ref{eq:RES1}) we get
\begin{align}
     K^{1,1}_{j}&= 
     (B_{j}^1)^{\frac{1}{2}} \nonumber \\
      K^{2,j}_{jk}&= \sqrt{{{A_{jk}}}}(B_{j}^1)^{-\frac{1}{2}}     
 \end{align} 
The above Kraus operators can be expressed using the Identity and Pauli matrices, which simplifies the calculation of powers of matrices (for a short reminder, see Appendix \ref{appendix2}). In this way, we obtain the following explicit expressions for the Kraus operators.
\begin{align}
   K^{1,1}_{j}&=F_{+}{I}+F_{-}(\tcb{\vec{v}_{j}}\cdot\vec{\sigma})  \\
    K^{2,j}_{jk}&=[G_{k,+}{I}+G_{k,-}(\tcb{\vec{n}_{jk}}\cdot\vec{\sigma})][H_{+}{I}+H_{-}(\tcb{\vec{v}_{j}}\cdot\vec{\sigma})] \nonumber
\end{align}
where the vectors 
\begin{align}
    \tcb{\vec{n}_{00}}&=(-\sin{\beta},0,-\cos{\beta}) &  \tcb{\vec{n}_{01}}&=(\sin{\beta},0,-\cos{\beta}) \nonumber\\
      \tcb{\vec{n}_{10}}&=(0,-\sin{\alpha},\cos{\alpha})&  \tcb{\vec{n}_{11}}&=(0,\sin{\alpha},\cos{\alpha}) \nonumber 
\end{align}
and $\tcb{\vec{v}_{j}}=\sum_{k}\vec{n}_{jk}$ for $j,k=0,1$ and,
\begin{align}
   F_{\pm}&=\frac{f_{+} \pm f_{-}}{2}&\quad  f_{\pm}&=\sqrt{\frac{1\mp\eta}{2}}\nonumber  \\
   G_{j,\pm}&=\frac{g_{j,+} \pm g_{j,-}}{2}&\quad  g_{j,\pm}&= \frac{1}{2}\sqrt{1\mp \eta \sec{\theta_{j}}}\nonumber \\
H_{\pm}&=\frac{h_{+} \pm h_{-}}{2}&\quad  h_{\pm}&=\frac{1}{f_{\pm}},\nonumber  \\
\end{align}
where $\theta_0=\beta$ and $\theta_1=\alpha$.
One can verify that Eq. (\ref{eq:exampplleeeeeXX}) is fullfilled and that Kraus operators $K^{2,j}_{jk}$ and $K^{1,1}_{j}$ define valid instruments $\mathcal{I}^{2,j}$, $\mathcal{I}^{1,1}$, respectively.

\section{Sequential realization of Positive Operator-Valued Measures} \label{POVM section}

Any POVM can be understood as a special case of a quantum instrument whose output space is one-dimensional and whose outcome space and input Hilbert space coincide with those of the POVM. In this way, all the results derived in this paper directly apply also to POVMs. Thus, we explicitly adapt the presented results to POVMs and discuss their relation to previous work. In analogy to the instrument case, we refer to the POVM to be realized as the \emph{Total POVM}. As a direct consequence of Theorem \ref{main theorem}, we obtain the following corollary.

\begin{corollary} \label{corollary 4}
Given a Total POVM  
$A$ $\in \mathcal{O}(\Omega,\mathcal{H})$ and a postprocessing relation $A\rightarrow B$ to another POVM $B$ $\in \mathcal{O}(\Omega_B,\mathcal{H})$, there exists 
a $2$-step adaptive sequence consisting of Lüders Instrument $\mathcal{J}\in Ins(\Omega_B,\mathcal{H},\mathcal{H})$ of POVM $B$
and a set of 
POVMs $C^{j} \in \mathcal{O}(\Omega,\mathcal{H)}$ $\forall j\in \Omega_B$ such that 
\begin{align} \label{povmins}
    Tr[\rho A_k]&=Tr[\sum_{j\in \Omega_B} \mathcal{J}_{j} (\rho) C^{j}_k] \nonumber\\
    &=Tr[\rho \sum_{j\in \Omega_B}  \sqrt{B_j} C^{j}_k \sqrt{B_j}]\quad   \quad  \forall k \in \Omega.
\end{align}
\end{corollary}

\begin{proof}
One dimensional Hilbert space, $\mathcal{H}_{1d}$, is isomorphic to space of Complex numbers, since all vectors are unique multiples of a chosen basis vector. For the same reason also operators on $\mathcal{H}_{1d}$ are isomorphic to complex numbers. Having this in mind, we can define minimal set of Kraus operators for POVM $A$, which can be understood as an instrument $\mathcal{T}\in Ins(\Omega,\mathcal{H},\complexn)$ as 
\begin{align}
    T'_{k,m}=\sqrt{\lambda^k_m}\bra{v^k_m} \quad \forall k\in \Omega,\; m=1,\ldots,r_k,
\end{align}
where we used spectral decomposition of POVM elements 
\begin{align}
\label{eq:spect_Ak}
    A_k=\sum_{m=1}^{r_k}\lambda^k_m\ket{v^k_m}\bra{v^k_m}.
\end{align}
Indeed,
\begin{align}
    \mathcal{T}_k(\rho)=\sum_m T'_{k,m}\rho T'^\dagger_{k,m}=Tr(\rho A_k) \quad \forall k\in \Omega,
\end{align}
which verifies that the probability of outcome $k$ is determined by POVM $A$.

Since $\mathcal{T}$ maps from $d$-dimensional Hilbert space to one dimensional Hilbert space, from proof of Theorem \ref{main theorem} it is obvious that for instruments $\mathcal{R}^j$ additional Kraus operators determined by $R'^j\in Ins(\Omega,(I_\mathcal{H}-\Pi_j)\mathcal{H},\complexn)$ will be needed (see Eq. (\ref{eq:def_RESIDUAL_2})). 
They will guarantee the validity of Eq. (\ref{eq:normRjkm}). By definition we have 
\begin{align}
\label{eq:ortoKraus}
    R'^j_k B_j =0, \quad \sum_{k,m} R'^{j \dagger}_{k,m} R'^j_{k,m}=I_\mathcal{H}-\Pi_j. 
\end{align}
Let us denote 
\begin{align}
    C'^j_k\equiv\sum_{m} R'^{j \dagger}_{k,m} R'^j_{k,m}, 
\end{align}
since $\sum_m R'^j_{k,m} \rho R'^{j \dagger}_{k,m} = Tr(\rho C'^j_k)$.
It is obvious that $C'^j_k\geq 0$, $\sum_{k\in \Omega} C'^j_k=I_\mathcal{H}-\Pi_j$. We remind that due to Eq. (\ref{def:Kkjm}) we have ${T}_{k,jm}= \sqrt{\nu_{kj}\;\lambda^k_m}\bra{v^k_m}$ and together with the above equations we obtain
\begin{align}
     \label{eq:povm_case}
    \sum_{j \in \Omega_B} &(\RI^{j}_k \,\circ \,\mathcal{J}_{j})(\rho)=   \nonumber 
     \\ &= \sum_{j,m} \nu_{kj}\lambda^k_m\bra{v^k_m} B_{j}^{-\frac{1}{2}}B_{j}^{\frac{1}{2}}\rho B_{j}^{\frac{1}{2}}B_{j}^{-\frac{1}{2}}\ket{v^k_m} \nonumber\\ 
&\phantom{=} + \sum_m R'^j_{k,m} B_{j}^{\frac{1}{2}}\rho B_{j}^{\frac{1}{2}} R'^{j \dagger}_{k,m} \\
  &= Tr(\rho \sum_j B_{j}^{\frac{1}{2}}\;C^j_k\;B_{j}^{\frac{1}{2}})        \nonumber
\end{align} 

where we used Eq. (\ref{eq:spect_Ak}), 
cyclic property of the trace and we defined
\begin{align}
\label{eq:def_Cjk}
    C^j_k\equiv B_{j}^{-\frac{1}{2}} \nu_{kj} A_k B_{j}^{-\frac{1}{2}} + C'^j_k.
\end{align}
Although, we do not need to prove validity of Eq. (\ref{biggy}) in the considered special case, it is easy to check it. Due to Eqs. (\ref{eq:ortoKraus}), (\ref{eq:def_Cjk}) 
Eq. (\ref{eq:povm_case}) can be written as
\begin{align} 
    \label{eq:povm_case_final}
    \sum_{j \in \Omega_B} (\RI^{j}_k \,\circ \,\mathcal{J}_{j})(\rho) 
    =& Tr(\rho \sum_j  \nu_{kj} \Pi_{j} A_k\Pi_{j})   \\
    =& Tr(\rho \sum_j \nu_{kj} A_k) = Tr(\rho A_k),   \nonumber
\end{align}
where due to Lemma \ref{lmm:supsum} Eq. (\ref{postproc_POVM}) implies $A_k\Pi_{j}=A_k$ and we also used stochasticity of matrix $\nu$. 
Finally, from Eq. (\ref{eq:def_Cjk}) it is clear that $C^j_k\geq0$ and one can check that $\forall j\in \Omega_B$ $\sum_{k\in \Omega} C^j_k=\Pi_j+(I_\mathcal{H}-\Pi_j)=I_\mathcal{H}$, where we used stochasticity of matrix $\nu_{k}$. Thus, $C^j\in\mathcal{O}(\Omega_k,\mathcal{H})$  is a valid POVM, which concludes the proof.
\end{proof}

Similarly, Corollaries \ref{corollary 1} and \ref{corollary 2} can be rewritten for POVMs, showing how an $N$-step sequential implementation for the Total POVM can be obtained. In particular, by restricting Corollary \ref{corollary 2} to POVMs with two outcome sets $\lambda_i$, one recovers the result of Andersson and Oi \cite{AndOi} for the implementation of POVMs by binary search trees. 

Another perspective on the obtained result for POVMs from Corollary \ref{corollary 4}, as well as on Theorems \ref{main theorem} - \ref{theorem3}, follows from considering the classical information obtained during the sequential implementation of instruments. We develop this approach in the next section.


\section{Sequential implementation vs. postprocessing and incompatibility of instruments}
\label{sec:rel_incomp_postproc}

The aim of this section is to put the original results of this paper into slightly wider context. While for the sequential implementation of a Total instrument intermediate classical outcomes of the ASI can be forgotten (or discarded) they can be also kept (or copied) and analyzed. They provide some information about the actions we performed on the tested quantum system. In particular, in the $2$-step scenario outcome $j$ of POVM $B$ from Theorem \ref{main theorem}  
is simultaneously obtained together with simultaneous realization of the Total instrument $\mathcal{T}$. This situation is exactly what a compatibility of a Total instrument $\mathcal{T}$ and POVM $B$ asks for (for exact mathematical statement see \cite{lm2}, Definition $5$ and consider joint instrument defined by $\mathcal{G} \in Ins(\Omega_B\times \Omega, \mathcal{H},\mathcal{K})$, $\mathcal{G}_{(j,k)}=\RI^{j}_k \circ \mathcal{J}_{j}$). In other words, any implementation of a Total Instrument via ASI implies that all the effective POVMs associated to the intermediate outcomes must be compatible with the Total instrument. This shows a close link between the problem of sequential implementation of instruments and incompatibility of instrument and a POVM. 
This relation can be described also starting from the incompatibility of instruments. Suppose one chooses a compatible pair of  instrument and a POVM. Corollary $3$ of \cite{lm2} states that instrument $\mathcal{T}$ and a POVM $B$ are compatible if and only if the instrument $\mathcal{T}$ can be obtained by instrument postprocessing ($\mathcal{T}\preceq \mathcal{I}^B$) from a Lüders instrument $\mathcal{I}^B$ defined by the POVM $B$. 
On the other hand, ASI was defined as a generalization of the instrument postprocessing to multi-step scenario. As we already stated in Remark \ref{remarkASI1} each $2$-step ASI is a particular realization of an instrument postprocessing relation of the ASI's Initial instrument to it's Total instrument $\mathcal{T}$.
Thus, Corollary $3$ of \cite{lm2} essentially says the following.

\begin{proposition}
Consider a Total Instrument $\mathcal{T}$ $\in Ins(\Omega,\mathcal{H},\mathcal{K})$ and a POVM $B \in \mathcal{O}(\Omega_B,\mathcal{H})$. Two step adaptive sequence constituted by an Initial Instrument $\mathcal{J}$ $\in Ins(\Omega_B,\mathcal{H},\mathcal{H})$ and a set of Residual Instruments $\RI^{j} \in Ins(\Omega,\mathcal{H},\mathcal{K})$ $\forall j \in \Omega_B$ such that
\begin{equation} \label{eq:totallll}
    \mathcal{T}_k(\rho)= \sum_{j\in \Omega_B} (\RI^{j}_k \circ \mathcal{J}_{j})(\rho)=\sum_{j\in \Omega_B} \RI^{j}_k (\sqrt{B_{j}}\rho\sqrt{B_{j}})  
    \end{equation}   
$\forall k\in \Omega$, exists if and only if  POVM $B$ is compatible with the Total instrument $\mathcal{T}$. 
\end{proposition}

More precisely, it means that the implementability of an instrument via $2$-step ASI is in one-to-one correspondence with the compatibility of that instrument and the POVM induced by the first step of such an ASI. At first sight, this result seems quite useful, but one should note that the compatibility of instruments is not yet sufficiently understood, and the methods developed in \cite{lm2} are not very practical for constructing the desired ASI.

In contrast, the goals of this paper were much more practical. We aimed at explicitly constructing quite wide set of ASI, which implement the given instrument and we aimed at minimizing the used ancilla dimension through out the implementation. From this viewpoint Theorem \ref{main theorem} and especially it's constructive proof delivers better starting point for the derivation of all the results of this paper. 

An instrument is always compatible with it's induced POVM. Since we can apply classical postprocessing on the obtained outcomes it is compatible also with any postprocessing of it's induced POVM. However, in general such POVMs will be only subset of the POVMs compatible with the given instrument. Thus, one might expect that Theorem \ref{main theorem} does not provide all possible $2$-step ASI realizing the given instrument while having Lüders Initial instrument. 
A simple example of this type arises in the implementation of POVMs. While in Corollary \ref{corollary 4} the POVM $B$ defining the first step of ASI must be a postprocessing of the Total POVM $A$, in Corollary $4$ from \cite{lm2} it can be any POVM $B$ compatible with POVM $A$. Although, the results presented in this paper describe only a subset of sequential implementation possibilities, they offer better control over the required resources and a clear way to extend the implementation methods to a multi-step scenario. For example, the possibility of sequential implementation of compatible POVMs was already described by Heinosaari and Miyadera \cite{univ}. However, to apply their methods to construct ASI, one must first verify whether the chosen initial POVM and the Total POVM are compatible, which is a separate nontrivial problem, often solvable only using (numerical) SDP techniques. Although our result for POVMs is less general, the allowed initial POVMs are much more easily parametrized (simply by choosing a postprocessing matrix), and the proof of our Corollary \ref{corollary 4} also provides simple ways to construct the corresponding ASI.

The second concept that is closely linked to the sequential implementation of instruments is instrument postprocessing. In \cite{lm1}, postprocessing of instruments was defined, and the primary focus was to identify instruments to which a given instrument can be postprocessed. In contrast, the aim of this paper can be seen as the search for all instruments that can be postprocessed to a given fixed instrument, or, in the multi-step case, to find sequences of postprocessing relations that yield the Total instrument. Similarly to the incompatibility of instruments, the topic of instrument postprocessing is not yet well developed. Nevertheless, there are some findings which we can relate to the problem under consideration.

In particular, if we rewrite Proposition $8$ from \cite{lm1} in the language of this paper then the authors prove that if the Total instrument $\mathcal{T}$ of a $2$step ASI is indecomposable, then the POVM induced by the initial instruments must be a postprocessing of the POVM $A^{\mathcal{T}}$. This shows that for indecomposable instruments, Theorem \ref{main theorem} already covers all possibilities for the POVM induced by the first step of the $2$-step ASI. It seems one could apply this reasoning recursively, as in Corollary \ref{corollary 2}. However, to do so rigorously and to formulate useful conclusions would require an amount of work which goes beyond the scope of this manuscript. Thus, exploration of this idea will be left for future work. 



\section{Summary and Discussion}
\label{sec:summary}

This manuscript focuses on optimizing space resources while implementing the desired quantum evolution. 
Such a task is relevant for the majority of current NISQ devices, since they implement quantum circuits in which mid-circuit measurements are becoming available and reasonably reliable. The presented work is theoretical and is based on a mathematical description of mid-circuit measurements as quantum instruments, capturing both the outcome statistics and the state change induced by such measurements. In particular, we define the notion of an $N$-step adaptive sequence of instruments (ASI), which describes any quantum evolution interleaved with $N$ measurements. Quantum instruments can have different dimensions of input and output Hilbert spaces, or different cardinalities of the outcome space, which allows them to be equivalent to quantum states, channels, or POVMs if these parameters are suitably chosen. Thus, the results we obtained are valid not only for sequential realizations of instruments, but also for sequential realizations of channels and POVMs. The central question is to find an ASI for a fixed Total instrument that we intend to realize efficiently. Theorem \ref{main theorem} and its proof constructively show how this can be done, starting from any postprocessing of the induced POVM of the Total instrument. 
In the lab, one still needs to implement each of the quantum instruments in the ASI by using ancillary systems, unitary transformations, and projective measurements on subsystems. Here, we assume that quantum systems can be reset and reused \cite{reuse}, which is currently becoming a widely available tool for circuit optimization in NISQ-era devices. We analyze the dimensionality of the ancillary systems needed in all branches and steps of the evolution, taking into account differences in the input/output dimensions. We find that coarse-grainings, i.e., postprocessings which simply merge outcomes, generally lead to smaller ancilla dimensions in applications of Theorem \ref{main theorem}. For Total instruments whose output Hilbert space dimension is not smaller than that of the input, in Theorem \ref{th:opt2stepsplit} we determine the minimal ancilla dimension for the $2$ step ASI\tcb{, which originate from application of Theorem \ref{main theorem}}. On the other hand, for Total instruments whose output Hilbert space dimension is smaller than that of the input, the situation is more complex. Nevertheless, we present Theorem \ref{theorem3}, which achieves an improvement in the dimension of the ancillary system used in this case, compared with Theorem \ref{main theorem}, by choosing a smaller intermediate Hilbert space connecting the two steps of the ASI. We generalize the results to the $N$ step scenario in Section \ref{sec:n-asi}. Interestingly, if the Total instrument has outcome space with a cartesian product structure then we show that  such $N$-step ASI exist, which determine in each step one of the elements $(a_1,\ldots, a_N)$  
of the final $N$ tuple outcome. This has an advantage of requiring less classical memory to store intermediate results, or earlier termination of the evolution if we are waiting only for some subset of results. 
Previously proposed schemes for sequential implementation of POVMs \cite{AndOi,danobouda,MCMexpt2024}, channels or instruments \cite{channelcons} can be seen as a special case of Corollary \ref{corollary 1}
proved in this paper. 
As a case study we also illustrate 
application of our Theorems for Lüders instruments defined by a three parametric class of $2$-qubit informationally complete POVMs. 

Particularly relevant in the NISQ era might be the result we managed to obtain 
for sequential implementations of instruments with non-shrinking Hilbert space dimension.  
 ASI, which are constructed using Corollary \ref{corollary 1} have due to 
 Corollary \ref{cor:tradeoff1} the number of mid-circuits measurements (steps of ASI) 
 times the number of used ancillary qubits lower bounded by a constant, which depends on the Kraus rank of the implemented instrument. This provides a practical way of trading off these two resources. Moreover, our considerations 
 about general limitations of $N$-step ASI 
 suggest that qualitatively this type of tradeoff is inevitable if the intermediate Hilbert space dimensions of the ASI can not be (sufficiently often) below the input dimension. 
Another interesting result is expressed by Corollary \ref{cor:low_mem_Inst_Ns}. In terms of qubits, we show that for any quantum instrument which transforms $n$ to $m(>n)$ qubits, there exist $N$-step ASI implementing it just with $(m-n)$ ancillary qubits, which are remeasured $(N-1)$ times and finally used as output qubits. 
In contrast, in previously published results on sequential decompositions of POVMs, channels/instruments 
increase of Hilbert space dimension was not considered and at least one qubit ancilla was needed in the studied cases.

The results of this manuscript could be seen as providing tools for performing tradeoff between number of qubits used (system plus ancilla) and the number of (intermediate) mid-circuit measurements. Presented methods were designed with the aim to provide constructive and practical solutions, while keeping the above resources under control. As we discuss in Section \ref{sec:rel_incomp_postproc} in general the problem of finding $2$-step ASI for a given Total instrument is related to compatibility of the Total instrument and a POVM. How this relation looks in the $N$-step case is an open question. Already for the $2$-step case presented results do not need to provide exhaustive / optimal possibilities for construction of ASI, since (except for indecomposable instruments) they explore only a subset of compatibility possibilities. Thus, future research could explore potential of these uncovered cases. Progress in that direction would benefit from further progress on incompatibility of quantum instruments. On the other hand, our results 
on general limitations of $N$-step ASI show that ultimate trade-off bounds could be constructed if we would have better understanding of achievable Hilbert space dimensions in the ASI for a given Total instrument. 
Another important parameter for NISQ devices is the number of $2$-qubit gates. 
Universal decompositions of state preparations, unitary transformations, and isometries into quantum gates could be used to implement instruments obtained by the methods developed in this paper. It is known that such decompositions in general require an exponential number of gates with respect to the number of qubits on which the transformation acts. However, transformations relevant for quantum computing must be implementable with a polynomial number of gates. Unfortunately, universal decompositions often yield exponentially long quantum circuits even for transformations known to be efficiently realizable. This is partly why it is difficult to extend our analysis from the required ancilla dimension to the required number of gates. Conclusions based on universal decompositions may hold for most random instances of the problem (such an approach to resource counting for sequential implementation was considered previously by Iten, Colbeck, and Christandl in \cite{ICC_2017}); however, in practice one is most likely interested in a set of measure zero, which is relevant for applications. In any case, exploring the workflow discussed above for medium system sizes and practically relevant instances of quantum channels, POVMs, or instruments might also be an interesting direction which could help clarify the potential of the presented methods for practical applications.  
Another useful generalization of the presented results, especially in the NISQ era, is the approximate realization of quantum instruments. For POVMs \cite{oz1,oz2,oz3} and, more recently, for instruments \cite{tava}, the problem was addressed in a single-step approach by developing dedicated methods. However, how to mathematically address this kind of question for instruments, using a multi-step approach, remains an open question.

\noindent\textit{Note added.}---
After finishing this work we became aware of a  related independent work~\cite{Matsui2025}, that also studies the problem of implementing quantum instruments with a limited number of qubits. 
Both papers attempt to provide a solid theoretical framework for optimization of needed resources for adaptive quantum circuits.

\acknowledgements 
The authors would like to thank M. Ziman for fruitful discussions around the topic. 
This work was supported by the Slovak Research and Development Agency through the project APVV-22-0570 (DeQHOST), by project VEGA 2/0164/25 (Quantum Structures) and by project Cost CA22113. 
M.S. was partly funded by the EU NextGenerationEU through the Recovery and Resilience Plan for Slovakia under the project No. 09I03-03-V04-00777. 

\section*{Author contributions}
M.S. conceived the original idea of the research problem. S.S. and M.S. jointly developed the theoretical framework, derived and proved the results, and analyzed the examples. Both authors contributed to writing the manuscript, and both reviewed and approved the final version. The authors used Instatext and Grammarly to improve the language, grammar, and clarity of the manuscript. The authors reviewed and edited all text and take full responsibility for the content of the paper. 

\appendix

\section{Calculation of matrix powers in $\mathcal{H}_2$} \label{appendix2}
Any positive Hermitian operator $X$ acting on two-dimensional Hilbert space (due to the isomorphism with $\mathbb{C}_2$), can be written in the operator basis formed by identity and Pauli matrices, i.e. 
\begin{eqnarray}
    X= \alpha I + \beta (\vec{n}\cdot\vec{\sigma}).
\end{eqnarray}

The operator $\vec{n}.\vec{\sigma}$ can be diagonalized as
\begin{equation*}
   \vec{n}.\vec{\sigma} = \ket{\vec{n}+}\bra{\vec{n}+}-
\ket{\vec{n}-}\bra{\vec{n}-}
\end{equation*}
where $\ket{\vec{n}+}$ and $\ket{\vec{n}-}$ are 
two orthogonal pure states corresponding to the two opposite points of a Bloch sphere.
Consequently, operator $X$ can be recast in the spectral form as
\begin{equation}
   X= \lambda_{\vec{n}+}\ket{\vec{n}+}\bra{\vec{n}+}+\lambda_{\vec{n}-}\ket{\vec{n}-}\bra{\vec{n}-},
\end{equation}
where $\lambda_{\vec{n}+}=\alpha+\beta$, $\lambda_{\vec{n}-}=\alpha-\beta$. 
Now, it is straightforward to calculate any powers of the operator $X$,
i.e., for any real number $\gamma$, 
\begin{equation}
    X^{\gamma}= \lambda_{\vec{n}+}^{\gamma}\ket{\vec{n}+}\bra{\vec{n}+}+\lambda_{\vec{n}-}^{\gamma}\ket{\vec{n}-}\bra{\vec{n}-}
\end{equation}
and the result can be again recast in the Pauli basis 
\begin{equation}
    X^{\gamma}=\frac{\lambda_{\vec{n}+}^{\gamma}+\lambda_{\vec{n}-}^{\gamma}}{2}I + \frac{\lambda_{\vec{n}+}^{\gamma}-\lambda_{\vec{n}-}^{\gamma}}{2}(\hat{n}\cdot\vec{\sigma})
\end{equation}

\section{Support of a sum of positive operators} \label{appendix1}
Let us denote the kernel of an operator  $C$ by, $ker(C)$.
\begin{lemma} \label{lemma:1} 
 Let  $D \in \mathcal{L(H)}$, be a positive semidefinite operator. Then,  $\braket{\psi|D|\psi}=0$ for a vector $\ket{\psi}\in \mathcal{H}$ if and only if  $\ket{\psi} \in ker(D)$.
\end{lemma}
\begin{proof}
Clearly, $\ket{\psi} \in ker(D)$ implies $\braket{\psi|D|\psi}=0$. For the converse implication, we consider spectral decomposition $D=\sum_{i=1}^{r_D} d_i \ket{v_i}\bra{v_i}$ of operator $D$, in which $d_i>0$ and vectors $\{\ket{v_i}\}_{i=1}^{r_D}$ are orthonormal. We can extend these vectors into an orthonormal basis $\{v_i\}_{i=1}^{d} \in \mathcal{H}$. Thus, an arbitrary pure state $\ket{\psi}\in \mathcal{H}$ can be decomposed as  $\ket{\psi}=\sum_{i=1}^d c_i\ket{v_i}=\sum_{i=1}^{r_D} c_i\ket{v_i}+ \ket{\psi'}$, and $1= \sum_{i=1}^{r_D}|c_i|^2 + \braket{\psi'|\psi'}$. Expectation value $\braket{\psi|D|\psi}=0$ implies $\sum_{i=1}^{r_D} |c_i|^2=0$ and consequently $c_i=0$ for $i=1,\ldots,r_D$, $\ket{\psi}=\ket{\psi'}$. Due to $\braket{a_i|\psi'}=0$ for $i=1,\ldots,r_D$ we have $D\ket{\psi}=0$. Thus, we proved $\ket{\psi} \in ker(D)$, which concludes the proof.
\end{proof}
  
\begin{lemma}
\label{lmm:supsum}
For two positive-semidefinite operators $A,B \in \mathcal{L(H)}$,
\begin{equation} \label{oppp}
 supp(A+B)=span(supp(A),supp(B))   
\end{equation}
    
\end{lemma}
\begin{proof}
    Let us expand $A$ and $B$ using their spectral decompositions.
    So we have, $A=\sum_i e_i \ket{w_i}\bra{w_i}$ and $B=\sum_j f_j \ket{u_j}\bra{u_j}$ where $e_i,f_j>0$. Consider a state $\ket{\psi} \in ker(A+B)$. From Lemma (\ref{lemma:1}), we infer that $\braket{\psi|A+B|\psi}=0$. Due to $A,B\geq 0$ this implies $\braket{\psi|A|\psi}=\braket{\psi|B|\psi}=0$. So $\ket{\psi} \in ker(A)$ as well as $\ket{\psi} \in ker(B)$. Thus $ker(A+B)=ker(A)\cap ker(B)$.
    
    As we can write $ker(A)=(span(\{w_i\}_{i=1}^{r_A}))^\perp$ and $ker(B)=(span(\{u_i\}_{i=1}^{r_B}))^{\perp}$, we conclude that any $\ket{\psi} \in ker(A+B)$ must be orthogonal to all vectors $\{w_i\}_{i=1}^{r_A}, \{u_i\}_{i=1}^{r_B}$ and consequently also to their linear combinations. This proves $ker(A+B)\subseteq (span(supp(A),supp(B)))^\perp$. On the other hand, for every  $\ket{\psi}\in(span(supp(A),supp(B)))^\perp$ we clearly have $(A+B)\ket{\psi}=0$, so  $\ket{\psi} \in ker(A+B)$ and we conclude  Eq. (\ref{oppp}) holds.    
\end{proof}

\end{document}